\documentclass[letter,11pt]{article}
\usepackage{float,graphicx,verbatim,fullpage,hyperref,amssymb,amsmath,amsthm,amsfonts,enumerate,multicol, xspace,comment,xcolor,thm-restate,url, cite}
\usepackage[T1]{fontenc}
\usepackage[margin=1in]{geometry}
\usepackage{algo}

\usepackage{subfig}
\usepackage[margin=1in]{geometry}
\newcommand*{\email}[1]{\texttt{#1}}
\usepackage{hyperref}

\newtheorem{theorem}{Theorem}[section]
\newtheorem{lemma}{Lemma}[section]

\newtheorem{proposition}{Proposition}[section]
\newtheorem{claim}{Claim}[section]
\newtheorem{remark}{Remark}[section]

\def\final{0}  
\def\iflong{\iffalse}
\ifnum\final=0  
\newcommand{\cnote}[1]{{\color{red}[{\tiny Chandra: \bf #1}]\marginpar{*}}}
\newcommand{\knote}[1]{{\color{red}[{\tiny Karthik: \bf #1}]\marginpar{\color{red}*}}}
\newcommand{\todo}[1]{{\color{red}[{\tiny TODO: \bf #1}]\marginpar{\color{red}*}}}
\else 
\newcommand{\cnote}[1]{}
\newcommand{\knote}[1]{}
\newcommand{\todo}[1]{}
\fi  

\newcommand{\R}{\mathbb{R}}
\newcommand{\E}{\mathbb{E}}

\newcommand{\prob}{\mathsf{Pr}}
\newcommand{\collection}{\mathcal{C}}
\newcommand{\opt}{OPT_k}
\newcommand{\cost}{\text{cost}}
\newcommand{\core}{C}
\newcommand{\countzero}{\lambda^0}
\newcommand{\countone}{\lambda^1}
\newcommand{\countQ}{\lambda^Q}
\newcommand{\countintersect}{\lambda^{W_0 \cap Q}}
\newcommand{\countlhs}{\ell}
\newcommand{\countrhs}{r}
\newcommand{\positives}{\text{Positives}}
\newcommand{\negatives}{\text{Negatives}}
\newcommand{\deltacard}{d}

\newcommand{\hkpart}{\textsc{Hypergraph-$k$-Partition}\xspace}
\newcommand{\hkcut}{\textsc{Hypergraph-$k$-Cut}\xspace}
\newcommand{\hcut}{\textsc{Hypergraph-MinCut}\xspace}
\newcommand{\gkcut}{\textsc{Graph-$k$-Cut}\xspace}
\newcommand{\gcut}{\textsc{Graph-MinCut}\xspace}
\newcommand{\submodkpartl}{\textsc{Submodular-$k$-Partition}\xspace}
\newcommand{\submodkpart}{\textsc{Submod-$k$-Part}\xspace} 

\newcommand{\symsubmodkpart}{\textsc{Sym-Submod-$k$-Part}\xspace}

\newcommand{\mypara}[1]{\medskip \noindent {\bf #1}}

\def\complement#1{\overline{#1}}
\def\set#1{\{ #1 \}}
\usepackage[utf8]{inputenc}

\title{Hypergraph $k$-cut for fixed $k$ in deterministic polynomial
  time\footnote{University of Illinois, Urbana-Champaign. Email:
    \email{\{karthe,chekuri\}@illinois.edu}. Supported in part by NSF grant CCF-1907937.
		}}
\author{Karthekeyan Chandrasekaran
\and Chandra Chekuri
}
\date{}

\begin{document}

\maketitle

\begin{abstract}
  We consider the \hkcut problem. The input consists of a hypergraph  $G=(V,E)$ with non-negative hyperedge-costs $c: E\rightarrow \R_+$ and a positive integer $k$. The objective is to find a least-cost subset $F\subseteq E$ such that the number of connected components in $G-F$ is at least $k$. An alternative formulation of the objective is to find a partition of $V$ into $k$ non-empty sets $V_1,V_2,\ldots,V_k$ so as to minimize the cost of the hyperedges that cross the partition. \gkcut, the special case of \hkcut obtained by restricting to graph inputs, has received considerable attention. Several different approaches lead to a polynomial-time algorithm for \gkcut when $k$ is fixed, starting with the work of Goldschmidt and Hochbaum (1988) \cite{GH88,GH94}. In contrast, it is only recently that a randomized polynomial time algorithm for \hkcut was developed \cite{CXY19} via a subtle generalization of Karger's random contraction approach for graphs. In this work, we develop the first deterministic polynomial time algorithm for \hkcut for all fixed $k$. We describe two algorithms both of which are based on a divide and conquer approach. The first algorithm is simpler and runs in $n^{O(k^2)}$ time while the second one runs in $n^{O(k)}$ time. Our proof relies on new structural results that allow for efficient recovery of the parts of an optimum $k$-partition by solving minimum $(S,T)$-terminal cuts. Our techniques give new insights even for \gkcut. 
\end{abstract}

\newpage
\setcounter{page}{1}

\section{Introduction}
\label{sec:intro}
A hypergraph $G=(V,E)$ consists of a finite set $V$ of vertices and a
finite set $E$ of hyperedges where each $e \in E$ is a subset of $V$.  In this work, we consider
the \hkcut problem, in particular when $k$ is a fixed
constant. The input to this problem consists of a hypergraph $G=(V,E)$
with non-negative hyperedge-costs $c: E\rightarrow \R_+$ and a positive integer $k$. The objective is to find a minimum-cost subset of hyperedges whose
removal results in at least $k$ connected components. 
An equivalent partitioning formulation turns out to be quite important. 
In this formulation, the
objective is to find a partition of $V$ into $k$ non-empty sets
$V_1,V_2,\ldots, V_k$ so as to minimize the cost of the hyperedges that
cross the partition. A hyperedge $e \in E$ crosses a partition
$(V_1,V_2,\ldots,V_k)$ if it has vertices in more than two parts,
that is, there exist distinct $i, j\in [k]$ 
such that
$e \cap V_i \neq \emptyset$ and $e \cap V_j \neq \emptyset$.

Cut and partitioning problems in graphs, hypergraphs, and related
structures including submodular functions are extensively studied in
algorithms and combinatorial optimization literature for their theoretical
importance and numerous applications. \hkcut is a problem that is of
inherent interest not only for its applications and simplicity but also because
of its close connections to a special case, namely in graphs, and to a
generalization, namely in submodular functions. For this reason the complexity
of \hkcut has been an intriguing open problem for several years with
some important recent progress. First we describe these closely
related problems and some prior work on them.

\mypara{\gkcut:} This is a special case of \hkcut where the input is a
graph instead of a hypergraph. When $k=2$, \gkcut is the global minimum
cut problem (\gcut) which is a fundamental and well-known problem.  It
is easy to see that \gcut can be solved in polynomial time via
reduction to min $s$-$t$ cuts but there is more structure in
$\gcut$, and this can be exploited to obtain faster deterministic and
randomized algorithms \cite{NI92,SW97,KS96,Karger00}. The complexity
of \gkcut for $k \ge 3$ has also been extensively investigated
with substantial recent work. 
Goldschmidt and Hochbaum (1988) \cite{GH88,GH94} showed that \gkcut is NP-Hard when $k$ is part of the
input and that it is polynomial-time solvable when $k$ is any fixed
constant (this is not obvious even for $k=3$). They used a
divide-and-conquer approach for \gkcut which resulted in an algorithm
with a running time of $n^{O(k^2)}$. We will describe the technical
aspects of this approach in more detail later. This approach
has been refined over several papers culminating in an algorithm of
Kamidoi, Yoshida, and Nagamochi \cite{KYN07} that ran in
$n^{(4 + o(1))k}$ time. Two very different approaches also give
polynomial-time algorithms for fixed $k$. The first approach is the random contraction approach of Karger that, via the improvement in Karger and Stein's work, led to a Monte Carlo randomized algorithm with a running time of $\tilde{O}(n^{2k-2})$; very recently Gupta, Lee, and Li \cite{GLL20-STOC} showed that the Karger-Stein algorithm in fact runs in $\hat{O}(n^k)$ time (where $\hat{O}(\cdot)$ hides
$2^{O(\text{ln ln }n)^2}$); $n^{(1-o(1))k}$ appears to be lower bound on the run-time via a reduction from the problem of finding a maximum-weight clique of size $k$ (see \cite{Li19}). 
The second approach is the tree packing approach which was introduced by Karger for \gcut. 
Thorup \cite{Th08} showed that tree
packings can also be used to obtain a polynomial-time algorithm for \gkcut.  His
algorithm runs in deterministic $n^{2k + O(1)}$ time; his approach was
clarified in \cite{CQX19} via an LP relaxation and this also resulted
in a slight improvement in the run-time and currently yields the
fastest deterministic algorithm. We defer discussion of approximation
algorithms for \gkcut when $k$ is part of the input to the related
work section.

\mypara{Submodular Partition Problems:} Graph and hypergraph cut
functions are submodular and one can view \gkcut and \hkcut as
special cases of a more general problem called \submodkpartl
(abbreviated to \submodkpart) that we define now. We recall that a
real-valued set function $f:2^V \rightarrow \R$ is submodular iff
$f(A \cap B) + f(A \cup B) \le f(A) + f(B)$ for all
$A, B \subseteq V$. Zhao, Nagamochi, and Ibaraki \cite{ZNI05} defined
\submodkpart as follows: given $f$ specified via a value oracle and a positive
integer $k$, the goal is to partition $V$ into non-empty sets $V_1,V_2,\ldots,V_k$ so as to
minimize $\sum_{i=1}^k f(V_i)$. A special case of \submodkpart is
\symsubmodkpart when $f$ is symmetric (that is
$f(A) = f(V\setminus A)$ for all $A \subseteq V$). It is not hard to
see that \gkcut is a special case of \symsubmodkpart. However, \hkcut
is not a special case of \symsubmodkpart even though the hypergraph
cut function is itself symmetric;\footnote{\symsubmodkpart when the input function $f$ is the cut function of a hypergraph is known as \hkpart in the literature \cite{ZNI05, OFN12}. We emphasize that the objective in \hkpart is different from the objective in \hkcut.} as observed in \cite{OFN12}, one can
reduce \hkcut to \submodkpart. \submodkpart and \symsubmodkpart are very general problems. For $k=2$, they can be solved in polynomial-time
via submodular function minimization. It is a very interesting open
problem to decide whether they admit polynomial-time algorithms for
all fixed $k$. Okumoto, Fukunaga, and Nagamochi \cite{OFN12} showed
that \submodkpart is polynomial-time solvable for $k=3$. They
generalized the work of Xiao \cite{Xi08} who showed that \hkcut
is polynomial-time solvable for $k=3$. Queyranne claimed, in 1999, a
polynomial-time algorithm for \symsubmodkpart when $k$ is fixed \cite{Q99},
however the claim was retracted subsequently. This is reported in
\cite{GQ} where it is also shown that
\symsubmodkpart has a polynomial-time algorithm for $k \le 4$.

\smallskip
\noindent 
\emph{Multiterminal variants:} We also mention that \gkcut, \hkcut, and
\submodkpart have natural variants involving separating specified
terminal vertices $s_1,s_2,\ldots,s_k$. These versions are NP-hard for $k\ge 3$. We discuss approximation algorithms for these problems in the related work section.

\mypara{\hkcut and main result:} The complexity of \hkcut for fixed $k$ has been
open since the work of Goldschmidt and Hochbaum for graphs (1988) \cite{GH88}. For $k=2$, this is the \hcut
problem and can be solved via reduction to min $s$-$t$ cuts in
directed graphs \cite{La73} or via other approaches that take
advantage of the submodularity structure of the hypergraph cut
function (see \cite{ChekuriX18} and references therein). For $k \ge 3$ and bounded rank hypergraphs, 
Fukunaga \cite{F10} generalized
Thorup's tree packing approach \cite{Th08} to solve \hkcut for fixed $k$ --- the run-time depends exponentially in the rank (rank is the maximum cardinality of a
hyperedge in the input hypergraph). It was also observed that Karger's random
contraction approach for graphs easily extends to give a randomized algorithm for bounded rank hypergraphs. As we noted earlier, Xiao
\cite{Xi08} obtained a polynomial-time algorithm for \hkcut when
$k=3$. In fairly recent work, Chandrasekaran, Xu, and Yu \cite{CXY19}
obtained the first randomized polynomial-time algorithm for \hkcut for
any fixed $k$; their Monte Carlo algorithm runs in
$\tilde{O}(pn^{2k-1})$ time where $p = \sum_{e \in E} |e|$ is the
representation size of the input hypergraph. Subsequently, Fox, Panigrahi, and Zhang \cite{FPZ19} improved the
randomized run-time to $\tilde{O}(mn^{2k-2})$, where $m$ is the number
of hyperedges in the input hypergraph.  Both these randomized
algorithms are based on random contraction of hyperedges and are
inspired partly by earlier work in \cite{GKP17} for \hcut.

The existence of a randomized algorithm for \hkcut raises
the question of the existence of a deterministic algorithm.  Random
contraction based algorithms do not lend themselves naturally to
derandomization. Perhaps, more pertinent is our interest in addressing
the complexity of \submodkpart. There is no natural random contraction
approach for this more general problem. For \gkcut, two distinct
approaches lead to deterministic algorithms and among these,  the tree
packing approach, like the random contraction approach, does not
appear to apply to \submodkpart. This leaves the divide and conquer
approach initiated in the paper of Goldschmidt and Hochbaum
\cite{GH88,GH94}. Is there a variant of this approach that works for
\hkcut and \submodkpart? 
%
We discovered certain structural properties of \hkcut (that do not hold for other submodular functions) to prove our main result stated below. 
\begin{theorem}
  \label{thm:main-intro}
  There is a deterministic polynomial-time algorithm for \hkcut for any fixed $k$.
\end{theorem}

Our work raises the hope for a polynomial-time algorithm for \submodkpart when $k$ is fixed.

\subsection{Technical overview and structural results}
We focus on the unit-cost variant of the problem in the rest of this
work for the sake of notational simplicity. Note that we allow
multigraphs and hence this is without loss of generality. All our
algorithms extend in a straightforward manner to arbitrary hyperedge
costs. They rely only on minimum $(s,t)$-cut computations and hence, they
are strongly polynomial.

A key algorithmic tool will be the use of terminal cuts.  We need some
notation. Let $G=(V,E)$ be a hypergraph. For a subset $U$ of vertices,
we will use $\complement{U}$ to denote $V\setminus U$, $\delta(U)$
to denote the set of hyperedges crossing $U$, 
and $\deltacard(U):=|\delta(U)|$ to denote the value of $U$.
More generally, given a partition
$(V_1,V_2,\ldots,V_h)$, we denote 
the number of hyperedges crossing the partition by 
$\cost(V_1, V_2, \ldots, V_h)$. 
Let $S$, $T$ be
disjoint subsets of vertices. A $2$-partition $(U, \complement{U})$ is
an $(S,T)$-terminal cut if $S\subseteq U\subseteq V\setminus T$. Here,
the set $U$ is known as the source set and the set $\complement{U}$ is
known as the sink set.  A minimum valued $(S,T)$-terminal cut is known
as a minimum $(S,T)$-terminal cut.  Since there could be multiple
minimum $(S,T)$-terminal cuts, we will be interested in \emph{source
  maximal} minimum $(S,T)$-terminal cuts and \emph{source minimal}
minimum $(S,T)$-terminal cuts.  These cuts are unique and can be found
in polynomial-time via standard maxflow algorithms. In fact, these
definitions extend to general submodular functions. Given $f:2^{V}\rightarrow \R$ and disjoint sets 
$S, T\subseteq V$, we can define a minimum $(S, T)$-terminal cut for $f$
as $\min_{U: S \subseteq U, T \subseteq \complement{U}} f(U)$. Uniqueness
of source-maximal and source-minimal $(S,T)$-terminal cuts follow from submodularity and one can also find these in polynomial-time via submodular function minimization.

Our algorithm follows the divide-and-conquer approach that was first
used by Goldschmidt and Hochbaum \cite{GH88,GH94} for \gkcut, and in a more general fashion by Kamidoi, Yoshida, and Nagamochi \cite{KYN07} to improve the running time for \gkcut.  The goal in this
approach is to identify one part of some fixed optimum $k$-partition
$(V_1,V_2,\ldots,V_k)$, say $V_1$ without loss of generality,
and then recursively find a $(k-1)$ partition of $V \setminus V_1$.
How do we find such a part? Goldschmidt and Hochbaum proved a key
structural lemma for \gkcut: Suppose $(V_1,V_2,\ldots,V_k)$ is
an optimum $k$-partition such that $V_1$ is the part with the smallest cut value (i.e., 
$|\delta(V_1)| \le |\delta(V_i)|$ for all $i\in [k]$) and
$V_1$ is maximal subject to this condition. Then, either
$|V_1| \le k-2$ or there exist disjoint sets $S, T$ such that
$S \subseteq V_1, T\subseteq \complement{V_1}$ with $|S|\le k-1$ and $|T\cap V_j|=1$ for every $j\in \{2,\ldots, k\}$ so that the source
maximal minimum $(S,T)$-terminal cut is $(V_1, \complement{V_1})$. One can guess/enumerate all
small-sized $(S,T)$-pairs to find an $O(n^{2k-2})$-sized collection of sets containing $V_1$ and recursively find an optimum $(k-1)$-partition of $V\setminus U$ for each $U$ in the collection.   
This leads to an $n^{O(k^2)}$-time algorithm for \gkcut.

Queyranne \cite{Q99} claimed that a natural generalization of the preceding structural lemma holds in the more general setting of \symsubmodkpart. However, as reported in \cite{GQ}, the claimed proof was incorrect and it was only proved for $k=3, 4$.
More importantly, as also noted in \cite{GQ}, this structural lemma (even if true for arbitrary $k$) is not useful for \symsubmodkpart 
because one cannot recurse on
$V\setminus V_1$; the function $f$ restricted to $V \setminus V_1$
is no longer symmetric! The reader might now wonder how the approach works for \gkcut?
Interestingly, \gkcut has the very nice property that the graph cut function restricted to $V\setminus V_1$ is still symmetric!

However, \hkcut, the problem of interest here, is \emph{not} a special case
of \symsubmodkpart.  Nevertheless, we are able to prove a strong
structural characterization. We state the structural characterization now. 
We consider the partition viewpoint of \hkcut. We will denote a $k$-partition by an ordered tuple. 
A $k$-partition is
a minimum $k$-partition if it has the minimum number of crossing
hyperedges among all possible $k$-partitions. Since there could be
multiple minimum $k$-partitions, we will be interested in the
$k$-partition $(V_1, \ldots, V_k)$ for which $V_1$ is maximal:
formally, we define a minimum $k$-partition $(V_1, \ldots, V_k)$ to be
a \emph{maximal minimum $k$-partition} if there is no other minimum
$k$-partition $(V_1', \ldots, V_k')$ such that $V_1$ is strictly
contained in $V_1'$. The following is our main structural result.

\begin{restatable}{theorem}{thmSmallWitness}
  \label{theorem:small-witness-inside-V_1-for-arbitrary-T}
  Let $G=(V,E)$ be a hypergraph and let $(V_1, \ldots, V_k)$ be a
  maximal minimum $k$-partition in $G$ for an integer $k\ge
  2$. Suppose $|V_1|\ge 2k-2$. Then, for every subset
  $T\subseteq \complement{V_1}$ such that $T$ intersects $V_j$ for
  every $j\in \set{2,\ldots, k}$,
  there exists a subset $S\subseteq V_1$ of size $2k-2$ such that
  $(V_1, \complement{V_1})$ is the source maximal minimum
  $(S,T)$-terminal cut.
\end{restatable}

Some important remarks regarding the preceding theorem are in order. 
Firstly, this is surprising: for instance, if the optimum $k$-partition is unique, then the theorem allows us to find any part $V_i$ of the optimum $k$-partition $(V_1, \ldots, V_k)$ by solving minimum $(S,T)$-terminal cuts for $S$ and $T$ of bounded sizes (by noting that the reordered $k$-partition $(V_i, V_1, \ldots, V_{i-1}, V_{i+1}, \ldots, V_k)$ is also a maximal minimum $k$-partition due to uniqueness and by applying Theorem \ref{theorem:small-witness-inside-V_1-for-arbitrary-T} to this reordered $k$-partition). Such a result was not known even for graphs. 
Secondly, our structural theorem differs crucially from the structural lemma of Goldschmidt and Hochbaum \cite{GH94} for
\gkcut in that it does not rely on $V_1$ being the part with the smallest cut value. 
This also explains why we need $S$ to be of size $2k-2$ instead of $k-1$: one can show that $2k-2$ is tight for our structural theorem if we want to identify an arbitrary part even when considering \gkcut. 
Thirdly, our structural theorem does not hold for
general submodular functions. 
The theorem statement was partly inspired by experiments on small sized instances and the proof is partly inspired by a structural theorem in \cite{KYN07} for graphs.

Theorem~\ref{theorem:small-witness-inside-V_1-for-arbitrary-T}
implies, relatively easily, an $n^{O(k^2)}$-time algorithm for
\hkcut. We improve the running time
to $n^{O(k)}$ using a similar but more involved structural result that allows us to recover the union of $k/2$ parts of an optimum $k$-partition. This high-level approach of recovering the union of $k/2$ parts of an optimum $k$-partition was developed in \cite{KYN07} for \gkcut. As we already mentioned
in the preceding paragraph, a proof of a key structural lemma in \cite{KYN07} was an inspiration for our proofs though the precise statement of our structural theorem is
different from the structural lemma of \cite{KYN07} and more subtle. 
We clarify this subtlety: the key structural lemma in \cite{KYN07} for graphs is that any $2$-partition whose cut value is strictly smaller than half the optimum $k$-cut value can be recovered as a minimum $(S,T)$-terminal cut for $S$ and $T$ of sizes at most $k-1$. In contrast, our structural theorem (Theorem \ref{theorem:small-witness-inside-V_1-for-arbitrary-T}) states that $V_1$---whose cut value need not necessarily be smaller than half the optimum $k$-cut value---can be recovered as a minimum $(S,T)$-terminal cut for $S$ and $T$ of sizes at most $2k-2$. We emphasize that the factor $2$ in the conclusion of our structural result (i.e., in the size of $S$) is not simply a consequence of weakening the hypothesis by a factor of $2$ compared to that of \cite{KYN07}. 

\medskip
\noindent
{\bf Organization.} In Section~\ref{section:recursive}, we formally
describe and analyze the basic recursive algorithm that utilizes our
main structural theorem (Theorem \ref{theorem:small-witness-inside-V_1-for-arbitrary-T}). We prove an important uncrossing property of
the hypergraph cut function in Section
\ref{section:uncrossing-for-hypergraph-cut-function} and use it to
prove Theorem \ref{theorem:small-witness-inside-V_1-for-arbitrary-T}
in Section \ref{section:structural-theorem-proof}. In
Section~\ref{section:structural-theorem-for-DC}, we 
prove 
a refined structural theorem and use it in Section \ref{section:DC-algo} to derive a faster algorithm based on divide-and-conquer. 

\subsection{Other related work}
Our main focus is on \hkcut and \gkcut when $k$ is fixed.  As we
mentioned already, \gkcut is NP-Hard when $k$ is part of the input
\cite{GH88}. A $2(1-1/k)$ approximation is known for \gkcut
\cite{SV95}; several other approaches also give a $2$-approximation
(see \cite{Q19,CQX19} and references therein). Manurangsi \cite{Ma17} showed
that there is no polynomial-time $(2-\epsilon)$-approximation for any
constant $\epsilon>0$ assuming the \emph{Small Set Expansion
  Hypothesis} \cite{RS10}. In contrast, \hkcut was recently shown
\cite{CL15} to be at least as hard as the \emph{densest
  $k$-subgraph} problem. Combined with results in \cite{Ma17dks}, this shows
that \hkcut is unlikely to have a sub-polynomial factor approximation
ratio and illustrates that \hkcut differs significantly from \gkcut
when $k$ is part of the input.

As we mentioned earlier, terminal versions of \submodkpart and its special
cases such as Multiway-Cut in graphs have been extensively studied. 
The most general version here is the following: given a submodular function $f:2^V \rightarrow \mathbb{R}$ (by value oracle) and
terminals $\{s_1,s_2,\ldots,s_k\} \subset V$ the goal is to find a partition $(V_1,\ldots,V_k)$ to minimize $\sum_i f(V_i)$ subject to the constraint that $s_i \in V_i$ for $1 \le i \le k$. These problems are NP-Hard even for $k=3$ and the main focus has been
on approximation algorithms.  We refer the reader to
\cite{BSW19,ZNI05,CE11,EneVW13} for further references. We mention that for
non-negative $f$ and fixed $k$, the best approximation algorithms for
\submodkpart and \symsubmodkpart are via the terminal versions; a
$(1.5-1/k)$ for \symsubmodkpart and a $2(1-1/k)$-approximation for
\submodkpart \cite{CE11,EneVW13}.

Fixed parameter tractability of \gkcut has also been investigated.  It
is known that \gkcut is $W[1]$-hard (and hence not likely to be FPT)
parameterized by $k$ \cite{DEFPR03} while it is FPT when parameterized
by $k$ and the solution size \cite{KT11}.  We observed, via a
simple reduction from a result of Marx on vertex separators
\cite{M06}, that \hkcut is $W[1]$ hard even when parameterized by $k$
\emph{and} the solution size. This also demonstrates that \hkcut differs in
complexity from \gkcut.

Another problem closely related to \hkcut is the \hkpart problem.  The
input to \hkpart is a hypergraph $G=(V,E)$ and a positive integer $k$ and the
goal is to partition $V$ into $k$ non-empty sets $V_1,\ldots,V_k$ but
the objective is to minimize $\sum_{i=1}^k|\delta_G(V_i)|$; this means that
a hyperedge $e$ that crosses $h \ge 2$ parts pays $h$ instead of only
once (as is the case in $\hkcut$).  \hkpart is a special case of \symsubmodkpart and
its complexity status for fixed $k \ge 5$ is open. \hkpart in constant rank hypergraphs is solvable in polynomial-time by relying on the fact that the number of constant-approximate minimum $k$-cuts in a constant rank hypergraph is polynomial.

\section{Recursive Algorithm}
\label{section:recursive}
Theorem \ref{theorem:small-witness-inside-V_1-for-arbitrary-T} allows
us to design a recursive algorithm for hypergraph $k$-cut that we
describe now. For a hypergraph $G=(V,E)$ and for a subset $U$ of
vertices, let $G[U]$ denote the hypergraph obtained from $G$ by
discarding the vertices in $\complement{U}$ and by discarding all
hyperedges $e\in E$ that intersect $\complement{U}$. We describe the
formal algorithm in Figure \ref{fig:k-cut-algorithm}. It follows the
high-level outline given in the technical overview. It enumerates
$n^{O(k)}$ minimum $(S,T)$-terminal cuts, one of which is guaranteed to
identify one part of an optimum $k$-partition, and then recursively
finds an optimum $(k-1)$-partition after removing the found part. The
run-time guarantee is given in Theorem \ref{theorem:k-cut-algorithm}.

\begin{figure*}[ht]
\centering\small
\begin{algorithm}
\textul{Algorithm CUT$(G,k)$}\+
\\  {\bf Input:} Hypergraph $G=(V,E)$ and an integer $k\ge 1$
\\  {\bf Output:} A $k$-partition corresponding to a minimum $k$-cut in $G$
\\  If $k=1$\+
\\      Return $V$\-
\\  else\+
\\      Initialize $\mathcal{C}\leftarrow \set{U\subset V: |U|\le 2k-3}$ and $\mathcal{R}\leftarrow \emptyset$
\\      For every disjoint $S, T\subset V$ with $|S|=2k-2$ and $|T|=k-1$\+
\\          Compute the source maximal minimum $(S,T)$-terminal cut $(U,\complement{U})$
\\          $\mathcal{C}\leftarrow \mathcal{C}\cup \set{U}$\-
\\      For each $U\in \mathcal{C}$\+
\\          $\mathcal{P}_{\complement{U}}:=\text{CUT}(G[\complement{U}],k-1)$
\\          $\mathcal{P}:=$ Partition of $V$ obtained by concatenating $U$ with $\mathcal{P}_{\complement{U}}$
\\          $\mathcal{R}\leftarrow \mathcal{R}\cup \set{\mathcal{P}}$\-
\\      Among all $k$-partitions in $\mathcal{R}$, pick the one with minimum cost and return it \-
\end{algorithm}
\caption{Algorithm to compute minimum $k$-cut in hypergraphs.}
\label{fig:k-cut-algorithm}
\end{figure*}

\begin{theorem}\label{theorem:k-cut-algorithm}
Let $G=(V,E)$ be a $n$-vertex hypergraph of size $p$ and let $k$ be an integer. Then, algorithm CUT$(G,k)$  returns a partition corresponding to a minimum $k$-cut in $G$ and it can be implemented to run in $n^{O(k^2)}T(n, p)$ time, where $T(n,p)$ denotes the time complexity for computing the source maximal minimum $(s,t)$-terminal cut in a $n$-vertex hypergraph of size $p$.
\end{theorem}
\begin{proof}
  We first show the correctness of the algorithm. All candidates
  considered by the algorithm correspond to a $k$-partition, so we
  only have to show that the algorithm returns a $k$-partition
  corresponding to a minimum $k$-cut. We show this by induction on
  $k$. The base case of $k=1$ is trivial. We show the induction
  step. Assume that $k\ge 2$. Let $(V_1, \ldots, V_k)$ be a maximal
  minimum $k$-partition with cost $\opt$. By Theorem
  \ref{theorem:small-witness-inside-V_1-for-arbitrary-T}, the
  $2$-partition $(V_1, \complement{V_1})$ is in $\mathcal{C}$. By
  induction hypothesis, the algorithm will return a minimum
  $(k-1)$-partition $(Q_1, \ldots, Q_{k-1})$ of
  $G[\complement{V_1}]$. Hence,
\[
  \cost_{G[\complement{V_1}]}(Q_1, \ldots, Q_{k-1})\le \cost_{G[\complement{V_1}]}(V_2, \ldots, V_{k}).
\]
Therefore, the cost of the $k$-partition $(V_1, Q_1, \ldots, Q_{k-1})$ is
\begin{align*}
  \deltacard(V_1) + \cost_{G[\complement{V_1}]}(Q_1, \ldots, Q_{k-1})
  &\le \deltacard(V_1) + \cost_{G[\complement{V_1}]}(V_2, \ldots, V_{k})
    = \opt.
\end{align*}
Moreover, the $k$-partition $(V_1, Q_1, \ldots, Q_{k-1})$ is in $\mathcal{R}$. Hence, the algorithm returns a $k$-partition with cost at most $\opt$.

Next, we bound the run-time of the algorithm. Let $N(k,n)$ denote the
number of source maximal minimum $(s,t)$-terminal cut computations
executed by the algorithm CUT$(G,k)$ on a $n$-vertex hypergraph
$G$. We note that $|\mathcal{R}|
=|\mathcal{C}|=O(n^{3k-3})$. Therefore,
\begin{align*}
  N(k,n) &\le O(n^{3k-3})(1+N(k-1, n)) \text{ and}\\
  N(1, n) &= O(1).
\end{align*}
Hence, $N(k,n)=O(n^{3k(k-1)/2})$. The total run-time is dominated by
the time to implement these minimum $(s,t)$-terminal cuts and hence it is
$O(n^{3k(k-1)/2})T(n,p)$.

\end{proof}


\section{Uncrossing properties of the hypergraph cut
  function}
\label{section:uncrossing-for-hypergraph-cut-function}

In this section, we show the following uncrossing theorem which will be
useful to prove the main structural theorem. See Figure
\ref{figure:uncrossing} for an illustration of the sets that appear in
the statement of Theorem \ref{theorem:hypergraph-uncrossing}.
The motivation for the statement of this uncrossing theorem will be clearer
in the proof of
Theorem~\ref{theorem:small-witness-inside-V_1-for-arbitrary-T}.
The reader may want to skip the rather
long and technical proof of the uncrossing theorem in the first
reading and come back to it after seeing its use in the proof of Theorem \ref{theorem:small-witness-inside-V_1-for-arbitrary-T}. 
\begin{restatable}{theorem}{thmHypergraphUncrossing}
\label{theorem:hypergraph-uncrossing}
Let $G=(V,E)$ be a hypergraph, $k\ge 2$ be an integer and
$\emptyset\neq R\subsetneq U\subsetneq V$. Let
$S=\{u_1,\ldots, u_p\}\subseteq U\setminus R$ for $p\ge 2k-2$. Let
$(\complement{A_i}, A_i)$ be a minimum
$((S\cup R)\setminus \set{u_i}, \complement{U})$-terminal cut. Suppose
that $u_i\in A_i\setminus (\cup_{j\in [p]\setminus \set{i}}A_j)$ for
every $i\in [p]$.  Then, there exists a $k$-partition
$(P_1, \ldots, P_k)$ of $V$ with $\complement{U}\subsetneq P_k$ such
that
\[
\cost(P_1, \ldots, P_k) \le \frac{1}{2}\min\{\deltacard(A_i) + \deltacard(A_j): i, j\in [p], i\neq j\}.
\]
\end{restatable}

\begin{figure}[htb]
\centering
\includegraphics[width=0.6\textwidth]{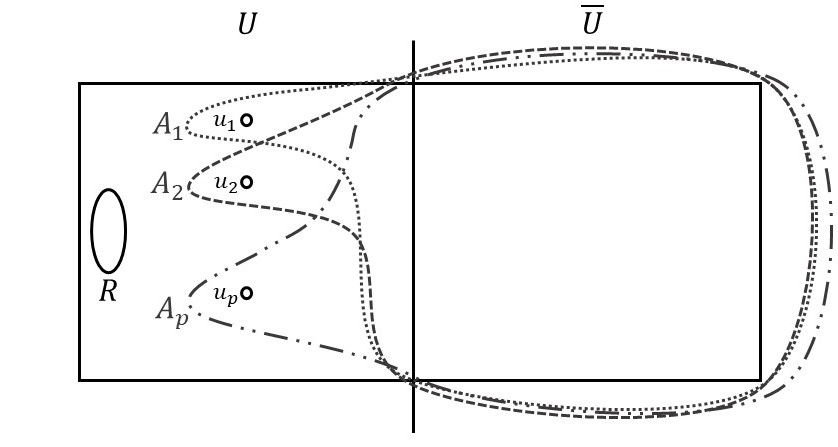}
\caption{Illustration of the sets that appear in Theorem \ref{theorem:hypergraph-uncrossing} and Lemma \ref{lemma:uncrossing}.}
\label{figure:uncrossing}
\end{figure}

The rest of the section is devoted to the proof of Theorem
\ref{theorem:hypergraph-uncrossing}.
We begin with some background on the hypergraph cut function.  Let
$G=(V,E)$ be a hypergraph. For a subset $A$ of vertices, we recall
that $\deltacard(A)$ denotes the number of hyperedges that intersect
both $A$ and $\complement{A}$. The function
$\deltacard: 2^V\rightarrow \R_+$ is known as the hypergraph cut
function. The hypergraph cut function is symmetric, i.e.,
\[
\deltacard(A) = \deltacard(\complement{A}) \text{ for all $A\subseteq V$},
\]
and submodular, i.e.,
\[
\deltacard(A) + \deltacard(B) \ge \deltacard(A \cap B) + \deltacard(A \cup B) \text{ for all subsets $A, B\subseteq V$}.
\]

For our purposes, it will help to count the hyperedges more accurately
than employ the submodularity inequality. We define some notation that
will help in more accurate counting.  Let $(Y_1, \ldots, Y_p, W, Z)$
be a partition of $V$. We recall that $\cost(Y_1, \ldots, Y_p, W, Z)$
denotes the number of hyperedges that cross the partition. We note
that when considering these hyperedges it is convenient to visualize
each part of the partition as a single vertex obtained by
contracting the part. 
We define
the following quantities:
\begin{enumerate}
\item Let
  $\cost(W, Z) = |\{ e \mid e \subseteq W \cup Z, e \cap W \neq \emptyset, e
  \cap Z \neq \emptyset\}|$ be the number of hyperedges contained in
  $W\cup Z$ that intersect both $W$ and $Z$.
    \item Let $\alpha(Y_1, \ldots, Y_p, W, Z)$ be the number of hyperedges that intersect $Z$ and at least two of the sets in $\set{Y_1, \ldots, Y_p, W}$.
    \item Let $\beta(Y_1, \ldots, Y_p, Z)$ be the number of hyperedges that are disjoint from $Z$ but intersect at least two of the sets in $\set{Y_1, \ldots, Y_p}$.
\end{enumerate}
For a partition $(Y_1, \ldots, Y_p, W, Z)$, we will be interested in
the sum of $\cost(Y_1, \ldots, Y_p, W, Z)$ with the three quantities
defined above which we denote as $\sigma(Y_1, \ldots, Y_p, W, Z)$,
i.e.,
\[
\sigma(Y_1, \ldots, Y_p, W, Z) := \cost(Y_1, \ldots, Y_p, W, Z) + \cost(W, Z) + \alpha(Y_1, \ldots, Y_p, W, Z) + \beta(Y_1, \ldots, Y_p, Z).
\]
We note that $\sigma(Y_1, \ldots, Y_p, W, Z)$ counts every hyperedge
that crosses the partition twice except for those hyperedges that
intersect exactly one of the sets in $\set{Y_1, \ldots, Y_p}$ and
exactly one of the sets in $\set{W, Z}$ which are counted exactly once
(see Figure \ref{figure:sigma-count}).

\begin{figure}[htb]
\centering
\includegraphics[width=0.5\textwidth]{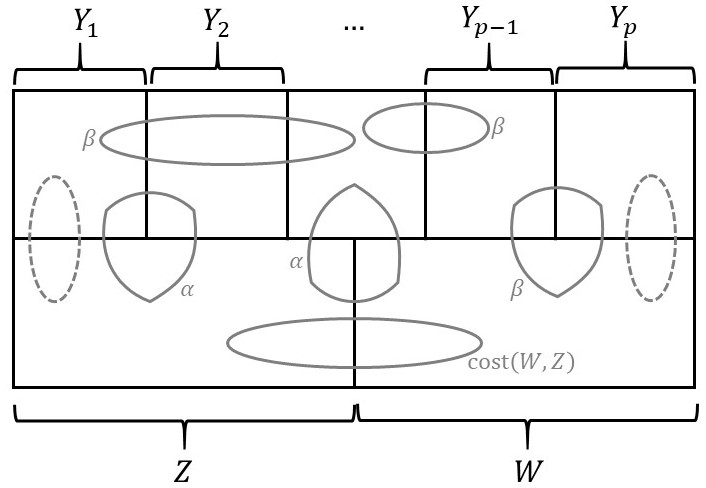}
\caption{Hyperedges counted by $\sigma(Y_1, \ldots, Y_p, W, Z)$: The dashed hyperedges are counted only by $\cost(Y_1, \ldots, Y_p, W, Z)$. The rest of the hyperedges are counted twice in $\sigma(Y_1, \ldots, Y_p, W, Z)$: once by the term $\cost(Y_1, \ldots, Y_p, W, Z)$ and once more by the indicated term.}
\label{figure:sigma-count}
\end{figure}

The motivation behind considering the function
$\sigma(Y_1, \ldots, Y_p, W, Z)$ comes from Proposition
\ref{prop:base-case}. We emphasize that the interpretation for
$\sigma(Y_1, \ldots, Y_p, W, Z)$ given in the proposition holds only
for $p=2$.

\begin{proposition}\label{prop:base-case}
  Let $(Y_1, Y_2, W, Z)$ be a partition of $V$ and let
  $A_1 := Y_1 \cup W$ and $A_2 := Y_2 \cup W$. Then,
  \[
    \deltacard(A_1) + \deltacard(A_2) = \sigma(Y_1, Y_2, W, Z).
  \]
\end{proposition}
\begin{proof}
  We show the equality by a counting argument. We prove that each
  hyperedge is counted the same number of times in LHS and RHS. We
  note that both LHS and RHS count only hyperedges that cross the
  partition $(Y_1, Y_2, W, Z)$. Let $e$ be a hyperedge that crosses
  the partition $(Y_1, Y_2, W,
  Z)$. Figure~\ref{figure:sigma-count-base-case} can be used to verify
  the equality. Formally we have the following cases:
\begin{enumerate}
\item Suppose $e$ intersects $Z$ and exactly one of the sets in $\set{Y_1, Y_2, W}$.
  \begin{enumerate}
  \item Suppose $e$ intersects $W$.
    Then, $e$ is counted twice in the LHS: by both $\deltacard(A_1)$ and $\deltacard(A_2)$. Moreover, $e$ is counted twice in the RHS: by $\cost(Y_1, Y_2, W, Z)$ and by $\cost(W, Z)$.
  \item Suppose $e$ intersects exactly one of the sets in $\set{Y_1, Y_2}$.
    Then, $e$ is counted once in the LHS: by exactly one of $\deltacard(A_1)$ and $\deltacard(A_2)$. Moreover, $e$ is counted exactly once in the RHS: by $\cost(Y_1, Y_2, W, Z)$.
  \end{enumerate}
\item Suppose $e$ intersects $Z$ and at least two of the sets in $\set{Y_1, Y_2, W}$.
  Then, $e$ is counted twice in the LHS: by both $\deltacard(A_1)$ and $\deltacard(A_2)$. Moreover, $e$ is counted twice in the RHS: by $\cost(Y_1, Y_2, W, Z)$ and by $\alpha(Y_1, Y_2, W, Z)$.
\item Suppose $e$ is disjoint from $Z$ and intersects both $Y_1$ and $Y_2$.
  Then, $e$ is counted twice in the LHS: by both $\deltacard(A_1)$ and $\deltacard(A_2)$. Moreover, $e$ is counted twice in the RHS: by $\cost(Y_1, Y_2, W, Z)$ and by $\beta(Y_1, Y_2, Z)$.
\item Suppose $e$ is disjoint from $Z$ and intersects exactly one of the sets in $\set{Y_1, Y_2}$. Since $e$ is crossing the partition $(Y_1, Y_2, W, Z)$, it has to  intersect $W$. Then, $e$ is counted once in the LHS: by exactly one of $\deltacard(A_1)$ and $\deltacard(A_2)$. Moreover, $e$ is counted exactly once in the RHS: by $\cost(Y_1, Y_2, W, Z)$.
\end{enumerate}
\begin{figure}[htb]
\centering
\includegraphics[width=0.5\textwidth]{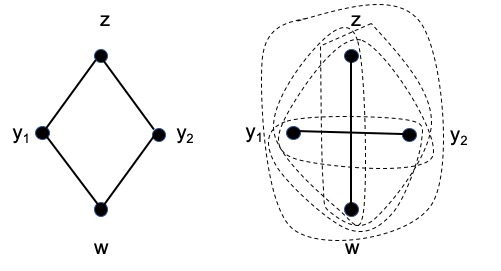}
\caption{Pictorial representation of hyperedges counted by
  $\sigma(Y_1, Y_2, W, Z)$. Contract each part to a single vertex.
  Figure on left shows hyperedges that are counted once and
  on the right all the rest that are counted twice; edges are shown as lines
  and hyperedges of size $\ge 3$ are shown in dashed lines. One can
  verify that only hyperedges that are counted once in $d(A_1) +
  d(A_2)$ correspond to precisely those in the left figure.}
\label{figure:sigma-count-base-case}
\end{figure}
\end{proof}

The next lemma will help in obtaining a $(p+3)$-partition from a
$(p+2)$-partition while controlling the increase in
$\sigma$-value. This will be used in a subsequent inductive argument.
See Figure \ref{figure:sigma-partition-set-uncross}
for an illustration of the sets appearing in the statement of the
lemma. 
Our proof of Lemma \ref{lemma:sigma-partition-set-uncross} is through case analysis. Currently we do not know how to prove this lemma without a somewhat
laborious case analysis. We remark that this is partly due to the fact
that hyperedges can have different cardinalities as well as due to the
fact that we cannot rely only on submodularity of the hypergraph cut
function. 

\begin{lemma}\label{lemma:sigma-partition-set-uncross}
Let $G=(V,E)$ be a hypergraph and let $(X_1, \ldots, X_p, W_0, Z_0)$ be a partition for some integer $p\ge 1$. Let $Q\subset V$ be a set such that
\[
Y_i:=X_i-Q\neq \emptyset\ \forall\ i\in [p],\ Y_{p+1}:=Q\cap Z_0 \neq \emptyset,\ Z:=Z_0-Q\neq \emptyset,\ \text{and}\ W:=W_0\cup(Q\setminus Z_0)\neq \emptyset.
\]
Then, $(Y_1, \ldots, Y_{p}, Y_{p+1}, W, Z)$ is a partition of $V$ such that
\[
\sigma(Y_1, \ldots, Y_{p}, Y_{p+1}, W, Z) \le \sigma(X_1, \ldots, X_{p}, W_0, Z_0) + \deltacard(Q) - \deltacard(W_0\cap Q).
\]
\end{lemma}

\begin{proof}
By definition $(Y_1, \ldots, Y_p, Y_{p+1}, W, Z)$ is a partition of $V$.

\begin{figure}[htb]
\centering
\includegraphics[width=0.6\textwidth]{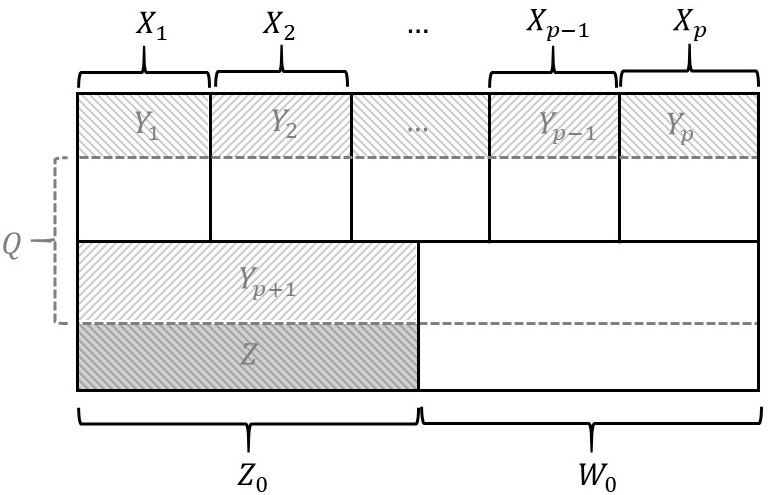}
\caption{Sets appearing in Lemma \ref{lemma:sigma-partition-set-uncross}. The unshaded portion corresponds to $W$.}
\label{figure:sigma-partition-set-uncross}
\end{figure}

We rewrite the required inequality in the following form as it becomes convenient to prove:
\begin{equation}
\sigma(X_1, \ldots, X_{p}, W_0, Z_0) - \sigma(Y_1, \ldots, Y_{p}, Y_{p+1}, W, Z)
\ge
\deltacard(W_0\cap Q) - \deltacard(Q). \label{ineq:to-show}
\end{equation}
For a hyperedge $e\in E$, let $\countzero_{e}\in \set{0,1,2}$ and $\countone_e\in \set{0,1,2}$ be the number of times that $e$ is counted by  $\sigma(X_1, \ldots, X_{p}, W_0, Z_0)$ and $\sigma(Y_1, \ldots, Y_{p}, Y_{p+1}, W, Z)$ respectively, and let $\countQ_e\in \set{0,1}$ and $\countintersect_e\in \set{0,1}$ be the number of times that $e$ is counted by $\deltacard(Q)$ and $\deltacard(W_0\cap Q)$ respectively.

Let $\countlhs_e:=\countzero_e - \countone_e$ and $\countrhs_e:=\countintersect_e - \countQ_e$. Thus, $\countlhs_e $ and $\countrhs_e$ denote the number of times the hyperedge $e$ is counted in the LHS and RHS of inequality (\ref{ineq:to-show}) respectively and moreover $\countlhs_e \in \set{0, \pm 1, \pm 2}$ and $\countrhs_e\in \set{0, \pm 1}$ for every hyperedge $e\in E$. Let
\begin{align*}
    \positives(\countlhs)&:=\sum_{e\in E: \countlhs_e\ge 1} \countlhs_e,\\
    \negatives(\countlhs)&:=\sum_{e\in E: \countlhs_e\le -1} \countlhs_e,\\
    \positives(\countrhs)&:=\sum_{e\in E: \countrhs_e= 1} \countrhs_e, \text{ and}\\
    \negatives(\countrhs)&:=\sum_{e\in E: \countrhs_e= -1} \countrhs_e.
\end{align*}

Claims \ref{claim:positives} and \ref{claim:negatives} complete the proof of the lemma.
\end{proof}

\begin{claim}\label{claim:positives}
\[
\positives(\countlhs)\ge \positives(\countrhs).
\]
\end{claim}
\begin{proof}
Let $e$ be a hyperedge such that $\countrhs_e=1$. Then, $e$ is counted by $\deltacard(W_0\cap Q)$ but not $\deltacard(Q)$. This means that $e\subseteq Q$, $e\cap (W_0\cap Q)\neq \emptyset$, and $e\cap (Q\setminus W_0)\neq \emptyset$. Thus, $e$ intersects $W_0\cap Q$ and at least one of the sets in $\set{X_1\cap Q, \ldots, X_p\cap Q, Z_0\cap Q}$. It suffices to show that $\countlhs_e\ge 1$. We consider different cases for $e$ below and show that $\countlhs_e\ge 1$ in all cases.

\begin{enumerate}
    \item Suppose $e$ intersects $Z_0\cap Q$.
    \begin{enumerate}
        \item Suppose $e$ is disjoint from $X_1\cap Q, \ldots, X_p \cap Q$.
        Then, $\countzero_e=2$ since $e$ is counted by both $\cost(X_1, \ldots, X_p, W_0, Z_0)$ and by $\cost(W_0, Z_0)$. However, $\countone_e=1$ since $e$ is counted only by $\cost(Y_1, \ldots, Y_{p+1}, W, Z)$. Hence, $\countlhs_e=\countzero_e-\countone_e\ge 1$.

        \item Suppose $e$ intersects at least one of the sets in $\set{X_1\cap Q, \ldots, X_p \cap Q}$.
        Then, $\countzero_e=2$ since $e$ is counted by both $\cost(X_1, \ldots, X_p, W_0, Z_0)$ and by $\alpha(X_1, \ldots, X_p, W_0, Z_0)$. However, $\countone_e=1$ since $e$ is counted only by $\cost(Y_1, \ldots, Y_{p+1}, W, Z)$. Hence, $\countlhs_e=\countzero_e-\countone_e\ge 1$.


    \end{enumerate}
    \item Suppose $e$ is disjoint from $Z_0\cap Q$. Then $e$ has to intersect at least one of the sets in $\set{X_1\cap Q, \ldots, X_p \cap Q}$.
    \begin{enumerate}
        \item Suppose $e$ intersects exactly one of the sets in $\set{X_1\cap Q, \ldots, X_p \cap Q}$.
        Then, $\countzero_e=1$ since $e$ is counted only by $\cost(X_1, \ldots, X_p, W_0, Z_0)$. However, $\countone_e=0$ since $e$ does not cross the partition $(Y_1, \ldots, Y_{p+1}, W, Z)$.       Hence, $\countlhs_e=\countzero_e-\countone_e\ge 1$.
        \item Suppose $e$ intersects at least two of the sets in $\set{X_1\cap Q, \ldots, X_p \cap Q}$.
        Then, $\countzero_e=2$ since $e$ is counted by both $\cost(X_1, \ldots, X_p, W_0, Z_0)$ and by $\beta(X_1, \ldots, X_p, Z_0)$. However, $\countone_e=0$ since $e$ does not cross the partition $(Y_1, \ldots, Y_{p+1}, W, Z)$. Hence, $\countlhs_e=\countzero_e-\countone_e=2\ge 1$.
    \end{enumerate}
\end{enumerate}
\end{proof}
\begin{claim}\label{claim:negatives}
\[
\negatives(\countlhs)\ge \negatives(\countrhs).
\]
\end{claim}
\begin{proof}
Let $e$ be a hyperedge such that $\countlhs_e\le -1$, i.e., $\countone_e\ge \countzero_e+1$. Then $\countone_e \ge 1$ and hence, $e$ crosses the partition $(Y_1, \ldots, Y_{p+1}, W, Z)$. It suffices to show that $\countrhs_e\le \ell_e$, i.e., $\countQ_e \ge \countintersect_e+\countone_e - \countzero_e$. We consider different cases for $e$ below and for each case, we show that either $\countQ_e \ge \countintersect_e+\countone_e - \countzero_e$ or the case is impossible.
\begin{enumerate}
    \item Suppose $e$ is disjoint from $Z$. Then, $e$ intersects at least one of the sets in $\set{Y_1, \ldots, Y_{p+1}}$ since $e$ crosses the partition $(Y_1, \ldots, Y_{p+1}, W, Z)$.
    \begin{enumerate}
        \item Suppose $e$ intersects exactly one of the sets in $\set{Y_1, \ldots, Y_{p+1}}$, say $Y_i$ for some $i\in [p+1]$.
        Then, $e$ intersects $W$ and consequently, $\countone_e=1$ since $e$ is counted only by $\cost(Y_1, \ldots, Y_{p+1}, W, Z)$.
        Since $1=\countone_e\ge \countzero_e+1$, it follows that $\countzero_e=0$. This implies that $e$ does not cross the partition $(X_1, \ldots, X_p, W_0, Z_0)$. Therefore, $i\in [p]$ and  $e\subseteq X_i$ with $e$ intersecting $X_i\cap Q$ and $Y_i=X_i\setminus Q$. Consequently, $\countQ_e=1$ and $\countintersect_e=0$. Hence $\countQ_e\ge \countintersect_e+\countone_e - \countzero_e$.

        \item Suppose $e$ intersects at least two of the sets in $\set{Y_1, \ldots, Y_{p+1}}$. Then, $\countone_e=2$ since $e$ is counted by both $\cost(Y_1, \ldots, Y_{p+1}, W, Z)$ as well as $\beta(Y_1, \ldots, Y_{p+1}, Z)$.

        \begin{enumerate}
            \item Suppose $e$ intersects at least two of the sets in $\set{Y_1, \ldots, Y_{p}}$.
            If $e$ intersects $Z_0$, then $\countzero_e=2$ since $e$ is counted by both $\cost(X_1, \ldots, X_p, W_0, Z_0)$ and $\alpha(X_1, \ldots, X_p, W_0, Z_0)$. 
            If $e$ is disjoint from $Z_0$, then again $\countzero_e=2$ since $e$ is counted by both $\cost(X_1, \ldots, X_p, W_0, Z_0)$ and $\beta(X_1, \ldots, X_p, W_0, Z_0)$. In both cases, we have $2=\countone_e\ge \countzero_e+1=3$, a contradiction.

            \item Suppose $e$ intersects $Y_{p+1}$ and exactly one of the sets in $\set{Y_1,\ldots, Y_p}$, say $Y_i$ for some $i\in [p]$. Then, $\countzero_e\ge 1$ since $e$ crosses the partition $(X_1, \ldots, X_p, W_0, Z_0)$. Since $2=\countone_e\ge \countzero_e+1$, it follows that $\countzero_e=1$. This implies that none of $\cost(W_0,Z_0)$, $\alpha(X_1, \ldots, X_p, W_0, Z_0)$, and $\beta(X_1, \ldots, X_p, Z_0)$ count $e$. Therefore, $e$ is disjoint from $W$ and $e$ intersects $Y_{p+1}=Z_0\cap Q$ and $Y_i=X_i\setminus Q$. Thus, $e$ is counted by $\deltacard(Q)$ but not $\deltacard(W_0\cap Q)$. Consequently, $\countQ_e=1$ and $\countintersect_e=0$. Hence, $\countQ_e\ge \countintersect_e+\countone_e - \countzero_e$.
        \end{enumerate}
    \end{enumerate}

    \item Suppose $e$ intersects $Z$. Then, $e$ intersects at least one of the sets in $\set{Y_1, \ldots, Y_{p+1}, W}$ since $e$ crosses the partition $(Y_1, \ldots, Y_{p+1}, W, Z)$.
    \begin{enumerate}
        \item Suppose $e$ intersects exactly one of the sets in $\set{Y_1, \ldots, Y_{p+1}, W}$. Then, $\countone_e=1$ since $e$ is counted only by $\cost(Y_1, \ldots, Y_{p+1}, W)$.

        \begin{enumerate}
            \item Suppose $e$ is disjoint from $W$. Then, $e$ intersects exactly one of the sets in $\set{Y_1, \ldots, Y_{p+1}}$. Since $1=\countone_e\ge \countzero_e+1$, we have that $\countzero_e=0$. This implies that $e$ does not cross the partition $(X_1, \ldots, X_p, W_0, Z_0)$. Hence, $e$ can only intersect $Y_{p+1}$. Thus, $e\subseteq Z_0=Z\cup Y_{p+1}$ with $e$ intersecting $Z=Z_0\setminus Q$ and $Y_{p+1}=Z_0\cap Q$. Thus, $e$ is counted by $\deltacard(Q)$ but not $\deltacard(W_0\cap Q)$. Consequently, $\countQ_e=1$ and $\countintersect_e=0$. Hence, $\countQ_e\ge \countintersect_e+\countone_e - \countzero_e$.

            \item Suppose $e$ intersects $W$.
            Then, $e$ has to cross the partition $(X_1, \ldots, X_p, W_0, Z_0)$ and therefore, $\countzero_e\ge 1$. Thus, $1=\countone_e\ge \countzero_e+1=2$, a contradiction.
        \end{enumerate}

        \item Suppose $e$ intersects at least two of the sets in $\set{Y_1, \ldots, Y_{p+1}, W}$. Then, $\countone_e=2$ since $e$ is counted by both $\cost(Y_1, \ldots, Y_{p+1}, W, Z)$ and $\alpha(Y_1, \ldots, Y_{p+1}, W, Z)$.

        \begin{enumerate}
            \item Suppose $e$ intersects at least two of the sets in $\set{Y_1, \ldots, Y_p}$.
            Then $\countzero_e=2$ since $e$ is counted by $\cost(X_1, \ldots, X_p, W_0, Z_0)$ as well as $\alpha(X_1, \ldots, X_p, W_0, Z_0)$. Thus, $2=\countone_e\ge \countzero_e+1 = 3$, a contradiction.

            \item Suppose $e$ intersects exactly one of the sets in $\set{Y_1, \ldots, Y_p}$, say $Y_i$ for some $i\in [p]$, and $e$ intersects $Y_{p+1}$ but is disjoint from $W$.
            Then, $\countzero_e\ge 1$ since $e$ crosses the partition $(X_1, \ldots, X_p, W_0, Z_0)$. Since $2=\countone_e\ge \countzero_e+1$, it follows that $\countzero_e=1$. This implies that none of $\cost(W_0, Z_0)$, $\alpha(X_1, \ldots, X_p, W_0, Z_0)$, and $\beta(X_1, \ldots, X_p, Z_0)$ count $e$ and hence, $e$ is contained in $Y_i\cup Z_0 \subseteq X_i\cup Z_0$ with $e$ intersecting $Y_{p+1}=Z_0\cap Q$ and $Y_i=X_i\setminus Q$. Thus, $e$ is counted by $\deltacard(Q)$ but not $\deltacard(W_0\cap Q)$. Consequently, $\countQ_e=1$ and $\countintersect_e=0$. Hence, $\countQ_e\ge \countintersect +\countone_e - \countzero_e$.

            \item Suppose $e$ intersects exactly one of the sets in $\set{Y_1, \ldots, Y_p}$, say $Y_i$ for some $i\in [p]$, and $e$ intersects $W$ but is disjoint from $Y_{p+1}$.
            Then, $\countzero_e\ge 1$ since $e$ crosses the partition $(X_1, \ldots, X_p, W_0, Z_0)$. Since $2=\countone_e\ge \countzero_e+1$, it follows that $\countzero_e=1$. This implies that none of $\cost(W_0, Z_0)$, $\alpha(X_1, \ldots, X_p, W_0, Z_0)$, and $\beta(X_1, \ldots, X_p, Z_0)$ count $e$. Therefore, $e$ is contained in $X_i\cup Z$ and $e$ intersects $X_i\cap Q$ since $e$ has to intersect $W$. Moreoever, $e$ intersects $Y_i=X_i\setminus Q$. Thus, $e$ is counted by $\deltacard(Q)$ but not $\deltacard(W_0\cap Q)$. Consequently, $\countQ_e=1$ and $\countintersect_e=0$. Hence, $\countQ_e\ge \countintersect_e+\countone_e - \countzero_e$.

            \item Suppose $e$ is disjoint from $Y_1, \ldots, Y_p$ and intersects both $Y_{p+1}$ and $W$.

            \begin{enumerate}
                \item Suppose $e$ intersects at least two of the sets in $\set{X_1\cap Q, \ldots, X_p\cap Q}$.
                Then, $\countzero_e=2$ since $e$ is counted by $\cost(X_1, \ldots, X_p, W_0, Z_0)$ as well as $\alpha(X_1, \ldots, X_p, W_0, Z_0)$. Thus, $2=\countone_e\ge \countzero_e+1=3$, a contradiction.

                \item Suppose $e$ does not intersect $X_1\cap Q, \ldots, X_p\cap Q$.
                Then, $e$ intersects $W_0$ since $e$ is counted by both $\cost(Y_1, \ldots, Y_{p+1}, W, Z)$ and $\alpha(Y_1, \ldots, Y_{p+1}, W, Z)$ (recall that we are in case (b)). Moreoever, $e\subseteq W_0\cup Z_0$. Therefore, $\countzero_e=2$ since $e$ is counted by $\cost(X_1, \ldots, X_p, W_0, Z_0)$ as well as $\cost(W_0, Z_0)$. Thus, $2=\countone_e\ge \countzero_e+1=3$, a contradiction.

                \item Suppose $e$ intersects exactly one of the sets in $\set{X_1\cap Q, \ldots, X_p\cap Q}$, say $X_i\cap Q$ for some $i\in [p]$, and $e$ intersects $W_0\cap Q$.
                Then, $\countzero_e=2$ since $e$ is counted by both $\cost(X_1, \ldots, X_p, W_0, Z_0)$ and $\alpha(X_1, \ldots, X_p, W_0, Z_0)$. Thus, $2=\countone_e\ge \countzero_e+1=3$, a contradiction.

                \item Suppose $e$ intersects exactly one of the sets in $\set{X_1\cap Q, \ldots, X_p\cap Q}$, say $X_i\cap Q$ for some $i\in [p]$, and $e$ is disjoint from $W_0\cap Q$. Then, $\countzero_e\ge 1$ since $e$ crosses the partition $(X_1, \ldots, X_p, W_0, Z_0)$. Since $2=\countone_e\ge \countzero_e+1$, it follows that $\countzero_e=1$. This implies that none of $\cost(W_0, Z_0)$, $\alpha(X_1, \ldots, X_p, W_0, Z_0)$, and $\beta(X_1, \ldots, X_p, Z_0)$ count $e$. Therefore, $e$ is contained in $(X_i\cap Q) \cup Z_0$ and $e$ intersects $Y_{p+1}=Z_0\cap Q$ and $Z=Z_0\setminus Q$. Thus, $e$ is counted by $\deltacard(Q)$ but not $\deltacard(W_0\cap Q)$. Consequently, $\countQ_e=1$ and $\countintersect_e=0$. Hence, $\countQ_e\ge \countintersect_e+\countone_e - \countzero_e$.

            \end{enumerate}
        \end{enumerate}
    \end{enumerate}
\end{enumerate}
\end{proof}

The next lemma will help in uncrossing a collection of sets to obtain
a partition with small $\sigma$-value. See Figure
\ref{figure:uncrossing} for an illustration of the sets that appear in
the statement of the lemma. 
\begin{lemma}\label{lemma:uncrossing}
  Let $G=(V,E)$ be a hypergraph and
  $\emptyset\neq R\subsetneq U\subsetneq V$. Let
  $S=\{u_1,\ldots, u_p\}\subseteq U\setminus R$ for $p\ge 2$. Let
  $(\complement{A_i}, A_i)$ be a minimum
  $((S\cup R)\setminus \set{u_i}, \complement{U})$-terminal
  cut. Suppose that
  $u_i\in A_i\setminus (\cup_{j\in [p]\setminus \set{i}}A_j)$ for
  every $i\in [p]$.  Let
\[
Z:= \cap_{i=1}^p \complement{A_i},\ W:= \cup_{1\le i<j\le p}(A_i \cap A_j),\ \text{and}\ Y_i:=A_i-W\ \forall i\in [p].
\]
Then, $(Y_1, \ldots, Y_p, W, Z)$ is a $(p+2)$-partition of $V$ with
\[
\sigma(Y_1, \ldots, Y_p, W, Z) \le \min\{\deltacard(A_i) + \deltacard(A_j): i, j\in [p], i\neq j\}.
\]
\end{lemma}
\begin{proof}
  For every $i\in [p]$, the set $Y_i$ is non-empty since $u_i\in
  Y_i$. The set $W$ is non-empty since $\complement{U}\subseteq
  W$. The set $Z$ is non-empty since $R\subseteq Z$. By definition,
  the sets $Y_1, \ldots, Y_p, W, Z$ are all disjoint and their union
  contains all vertices. Hence, $(Y_1, \ldots, Y_p, W, Z)$ is a
  partition of $V$.  Without loss of generality, let
  $\deltacard(A_1)\le \deltacard(A_2)\le \ldots \le \deltacard(A_p)$.
  We bound the $\sigma$-value of the partition by induction on $p$.


The base case of $p=2$ follows from Proposition \ref{prop:base-case}.
We show the induction step. Suppose that the statement holds for
$p=q$. We prove that it holds for $p=q+1$. Consider
$R_0:=R\cup \set{u_{q+1}}$ and $S_0:=S\setminus \set{u_{q+1}}$. Then,
$(\overline{A_i}, A_i)$ is still a minimum
$((S_0\cup R_0)\setminus \set{u_i}, \complement{U})$-terminal cut for
every $i\in [q]$ and moreover,
$u_i\in A_i\setminus \cup_{j\in [q]\setminus \set{i}}A_j$ for every
$i\in [q]$. By induction hypothesis, we get that for the sets
\[
Z_0:= \cap_{i=1}^q \complement{A_i},\ W_0:= \cup_{1\le i<j\le q}(A_i \cap A_j),\ \text{and}\ X_i:=A_i-W\ \forall i\in [q],
\]
we have
\[
\sigma(X_1, \ldots, X_q, W_0, Z_0) \le \deltacard(A_1) + \deltacard(A_2).
\]

The partition $(X_1, \ldots, X_q, W_0, Z_0)$ and the set $Q:=A_{q+1}$
satisfy the conditions of Lemma
\ref{lemma:sigma-partition-set-uncross}. By Lemma
\ref{lemma:sigma-partition-set-uncross}, we obtain that
\[
\sigma(Y_1, \ldots, Y_{q}, Y_{q+1}, W, Z) \le \sigma(X_1, \ldots, X_{q}, W_0, Z_0) +\deltacard(A_{q+1}) - \deltacard(W_0\cap A_{q+1}).
\]
Since $(\complement{W_0\cap A_{q+1}}, W_0\cap A_{q+1})$ is a feasible $((S\cup R)\setminus \set{u_{q+1}}, \complement{U})$-terminal cut, we have that $\deltacard(A_{q+1})\le \deltacard(W_0\cap A_{q+1})$. Hence,
\begin{align*}
\sigma(Y_1, \ldots, Y_{q}, Y_{q+1}, W, Z)
&\le \sigma(X_1, \ldots, X_{q}, W_0, Z_0)
\le \deltacard(A_1) + \deltacard(A_2).
\end{align*}
\end{proof}

The next lemma will help in aggregating the parts of a $2k$-partition $\mathcal{P}$ to a $k$-partition $\mathcal{K}$ so that the cost of $\mathcal{K}$ is at most half the $\sigma$-value of $\mathcal{P}$.
\begin{restatable}{lemma}{lemmaRecoverkCutFrompCut}
\label{lemma:recover-k-cut-from-p-cut}
Let $G=(V,E)$ be a hypergraph, $k\ge 2$ be an integer, and $(Y_1, \ldots, Y_p, W, Z)$ be a partition of $V$ for some integer $p\ge 2k-2$. Then, there exist distinct $i_1, \ldots, i_{k-1}\in [p]$ such that
\[
2\cost\left(Y_{i_1}, \ldots, Y_{i_{k-1}}, V\setminus (\cup_{j=1}^{k-1}Y_{i_j})\right)
\le
\cost(Y_1, \ldots, Y_p, W, Z) + \alpha(Y_1, \ldots, Y_p, W, Z) + \beta(Y_1, \ldots, Y_p, Z).
\]
\end{restatable}
\begin{proof}
  Suppose that the lemma is false. Pick a counterexample hypergraph
  $G=(V,E)$ such that $|V|+|E|$ is minimum. Hence, for every distinct
  $i_1, \ldots, i_{k-1}\in [p]$, we have
\[
2\cost\left(Y_{i_1}, \ldots, Y_{i_{k-1}}, V\setminus (\cup_{j=1}^{k-1}Y_{i_j})\right)
>
\cost(Y_1, \ldots, Y_p, W, Z) + \alpha(Y_1, \ldots, Y_p, W, Z) + \beta(Y_1, \ldots, Y_p, Z).
\]

Minimality of the counterexample implies that $|Y_i|=1$ for every
$i\in [p]$ and $|W|=1=|Z|$ (otherwise, we can obtain a smaller
counterexample by contracting the corresponding subset).  If there
exists a hyperedge $e\subseteq W\cup Z$ with $e$ intersecting both $W$
and $Z$, then discarding $e$ would still preserve the counterexample
property since $e$ is not counted in LHS but is counted in RHS, hence
no such hyperedge exists in $G$. For similar reasons, if there exists a
hyperedge $e$ that is double counted by RHS (see Figure
\ref{figure:sigma-count}), then discarding this hyperedge would still
preserve the counterexample property. Minimality of the counterexample
implies that no such hyperedge can exist. Consequently, all hyperedges
present in the hypergraph $G$ are in fact edges with one end-vertex in
$Y_i$ for some $i\in [p]$ and another end-vertex in $W$ or $Z$.
Thus,
\begin{align*}
    RHS = \cost(Y_1, \ldots, Y_p, W, Z) = \sum_{i=1}^{p}\deltacard(Y_i).
\end{align*}

Without loss of generality, let $\deltacard(Y_1)\le \deltacard(Y_2)\le \ldots \le \deltacard(Y_p)$. Then,
\begin{align*}
    2\cost\left(Y_{1}, \ldots, Y_{{k-1}}, V\setminus (\cup_{j=1}^{k-1}Y_{i_j})\right)
    &= 2\sum_{i=1}^{k-1} \deltacard(Y_i)
    \le \sum_{i=1}^p \deltacard(Y_i)
    = RHS.
\end{align*}
The inequality above is because $p\ge 2(k-1)$. Thus, $G$ cannot be a counterexample.
\end{proof}

We now restate and prove the main uncrossing theorem of this section.
\thmHypergraphUncrossing*
\begin{proof}
By applying Lemma \ref{lemma:uncrossing}, we obtain a $(p+2)$-partition $(Y_1, \ldots, Y_p, W, Z)$ such that
\[
\sigma(Y_1, \ldots, Y_p, W, Z) \le \min\{\deltacard(A_i) + \deltacard(A_j): i, j\in [p], i\neq j\}
\]
and moreover, $\complement{U}\subseteq W$. We recall that $p\ge 2k-2$. Hence, by applying Lemma \ref{lemma:recover-k-cut-from-p-cut} to the $(p+2)$-partition $(Y_1, \ldots, Y_p, W, Z)$, we obtain a $k$-partition $(P_1, \ldots, P_k)$ of $V$ such that $W\cup Z\subseteq P_k$ and
\[
\cost(P_1, \ldots, P_k) \le \frac{1}{2}\sigma(Y_1, \ldots, Y_p, W, Z) \le \frac{1}{2}\min\{\deltacard(A_i) + \deltacard(A_j): i, j\in [p], i\neq j\}.
\]
We note that $\complement{U}$ is strictly contained in $P_k$ since $\complement{U}\cup Z\subseteq W\cup Z\subseteq P_k$ and $Z$ is non-empty.
\end{proof}

\begin{remark}
The lower bound condition on $p$ (i.e., $p\ge 2k-2$) in the statement of Theorem \ref{theorem:hypergraph-uncrossing} is tight. In particular, the conclusion of the theorem does not hold for $p=2k-3$ as illustrated by the graph in Figure \ref{figure:tight-example}. \end{remark}

\begin{figure}[htb]
\centering
\includegraphics[width=0.4\textwidth]{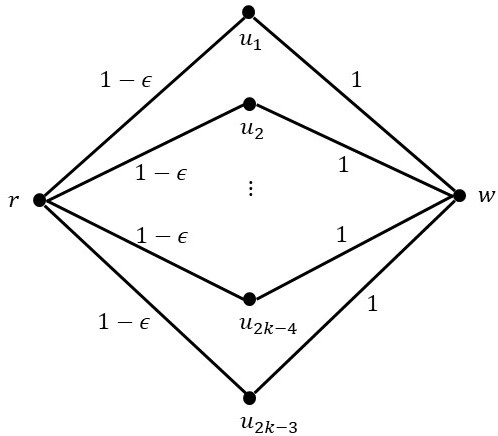}
\caption{An edge-weighted graph showing the necessity of the condition $p\ge 2k-2$ in Theorem \ref{theorem:hypergraph-uncrossing} (where $\epsilon$ is a small positive constant). We consider $U=\set{r, u_1, u_2, \ldots, u_{2k-3}}$ and $R={r}$. Then, the RHS of the theorem is $2k-3-\epsilon$ while the cost of any $k$-cut is at least $2k-2-O(\epsilon)$.}
\label{figure:tight-example}
\end{figure}

\begin{remark}
A natural counterpart of Theorem \ref{theorem:hypergraph-uncrossing} for (symmetric) submodular functions is false. For a submodular function $f:2^{V}\rightarrow \R_+$, by defining $f_{\text{sym}}(U):=f(U)+f(\complement{U})$ to be the value of the $2$-partition $(U, \complement{U})$, and assuming the conditions of the theorem, it is tempting to conjecture that there exists a $k$-partition $(P_1, \ldots, P_k)$ such that
\[
\sum_{i=1}^{k} f(P_i) \le \frac{1}{2}\min\left\{f_{\text{sym}}(A_i) + f_{\text{sym}}(A_j): i, j\in [p], i\neq j\right\}.
\]
Here is a counterexample: Consider the function $f(S):=1$ if $\emptyset\neq S\subsetneq V$, $f(\emptyset):=0$, and $f(V):=0$. Then, for any $k$-partition $(P_1, \ldots, P_k)$, we have $\sum_{i=1}^k f(P_i) = k$. However, the RHS in the above inequality is only $2$.
\end{remark}

\section{Proof of Theorem
  \ref{theorem:small-witness-inside-V_1-for-arbitrary-T}} \label{section:structural-theorem-proof}
In this section, we prove Theorem~\ref{theorem:small-witness-inside-V_1-for-arbitrary-T}.
We start with a useful containment property captured by the next lemma.

\begin{lemma}\label{lemma:V_S-inside-V_1}
  Let $G=(V,E)$ be a hypergraph, $(V_1, \ldots, V_k)$ be a maximal
  minimum $k$-partition in $G$ for an integer $k\ge 2$, and
  $S\subseteq V_1$, $T\subseteq \complement{V_1}$ such that
  $T\cap V_j\neq \emptyset$ for every $j\in \set{2,\ldots,
    k}$. Suppose $(U,\complement{U})$ is a minimum $(S,T)$-terminal
  cut. Then, $U\subseteq V_1$. \end{lemma}
\begin{proof}
  For the sake of contradiction, suppose
  $U\setminus V_1\neq \emptyset$.
  We will obtain another minimum
  $k$-partition that will contradict the maximality of $V_1$ in the
  minimum $k$-partition $(V_1, \ldots, V_k)$.  We observe that
\begin{align}
    \deltacard(U)&\le \deltacard(U\cap V_1) \label{ineq:intersection}
\end{align}
since $(U\cap V_1, \complement{U\cap V_1})$ is a $(S,T)$-terminal cut. We need the following claim:
\begin{claim}\label{claim:union}
\[
\deltacard(V_1)\le \deltacard(U\cup V_1).
\]
\end{claim}
\begin{proof}
  For the sake of contradiction, suppose
  $\deltacard(U\cup V_1)<\deltacard(V_1)$. Then, consider
  $W_1 := U\cup V_1$ and $W_j:= V_j\setminus U$ for every
  $j\in \set{2,\ldots, k}$ (see Figure
  \ref{figure:uncrossing-for-V_S-inside-V_1}).  We have
  $\deltacard(W_1)<\deltacard(V_1)$.  Since $S\subseteq W_1$ and
  $T\cap W_j\neq \emptyset$ for every $j\in \set{2,\ldots, k}$, we
  have that $(W_1, \ldots, W_k)$ is a $k$-partition. We will show that
  $\cost(W_1, \ldots, W_k)$ is strictly smaller than
  $\cost(V_1, \ldots, V_k)$, thus contradicting the optimality of the
  $k$-partition $(V_1, \ldots, V_k)$.

\begin{figure}[htb]
\centering
\includegraphics[width=0.5\textwidth]{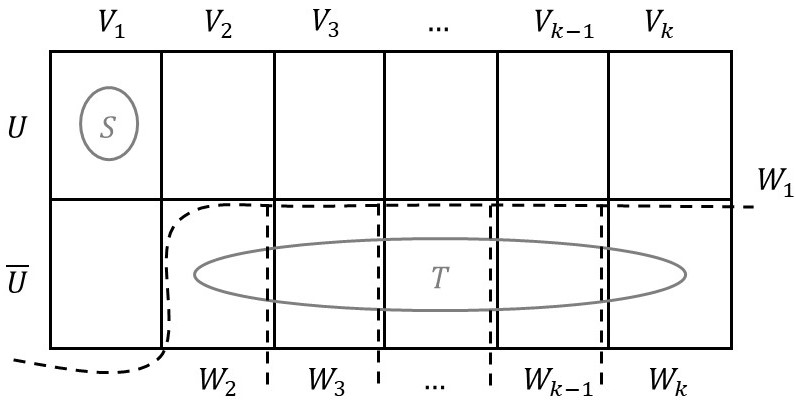}
\caption{Uncrossing in the proof of Claim \ref{claim:union}. }
\label{figure:uncrossing-for-V_S-inside-V_1}
\end{figure}

We recall that for a subset $A$ of vertices, the graph $G[A]$ is
obtained from $G$ by discarding the vertices in $\complement{A}$ and
by discarding the hyperedges that intersect $\complement{A}$. With
this notation, we can write
\begin{align*}
\cost_{G}(W_1, \ldots, W_k) &= \deltacard(W_1) + \cost_{G[\complement{W_1}]}(W_2, \ldots, W_k)\text{ and}\\
\cost_{G}(V_1, \ldots, V_k) &= \deltacard(V_1) + \cost_{G[\complement{V_1}]}(V_2, \ldots, V_k).
\end{align*}
Moreover, every hyperedge that is disjoint from $W_1=U\cup V_1$ but
crosses the $(k-1)$-partition
$(W_2=V_2\setminus U, \ldots, W_k=V_k\setminus U)$ is also disjoint
from $V_1$ but crosses the $(k-1)$-partition $(V_2, \ldots,
V_k)$. Hence,
$\cost_{G[\complement{W_1}]}(W_2, \ldots, W_k)\le
\cost_{G[\complement{V_1}]}(V_2, \ldots, V_k)$.  
We also have $\deltacard(W_1)< \deltacard(V_1)$. Therefore,
\[
\cost(W_1, \ldots, W_k) < \cost(V_1, \ldots, V_k),
\]
a contradiction to optimality of the $k$-partition $(V_1, \ldots, V_k)$.
\end{proof}

By inequality (\ref{ineq:intersection}), Claim \ref{claim:union}, and submodularity of the hypergraph cut function, we have that
\[
\deltacard(U) + \deltacard(V_1) \le \deltacard(U\cap V_1) + \deltacard(U\cup V_1) \le \deltacard(U) + \deltacard(V_1).
\]
Therefore, the inequality in Claim \ref{claim:union} should in fact be an equation, i.e.,
\[
\deltacard(V_1)=\deltacard(U\cup V_1).
\]

Going through the proof of Claim \ref{claim:union} with this additional fact, we obtain that the $k$-partition $(U\cup V_1, V_2\setminus U, \ldots, V_k\setminus U)$ has cost at most that of $(V_1, \ldots, V_k)$. Hence, the $k$-partition $(U\cup V_1, V_2\setminus U, \ldots, V_k\setminus U)$ is also a minimum $k$-partition and it contradicts the maximality of $V_1$.


\end{proof}

\begin{remark}
Lemma \ref{lemma:V_S-inside-V_1} also holds
for \submodkpart. That is, for a submodular
function $f:2^{V}\rightarrow \R_+$ with $(V_1, \ldots, V_k)$
being a maximal minimum $k$-partition for an integer $k\ge 2$,
subsets $S\subseteq V_1$ and $T\subseteq \complement{V_1}$ such that
$T\cap V_j\neq \emptyset$ for every $j\in \set{2,\ldots, k}$,
and $(U,\complement{U})$ being an $S,T$-separating $2$-partition with
minimum $f(U)+f(\complement{U})$ among all $S,T$-separating
$2$-partitions,
we have that $U\subseteq V_1$. This can be shown using the proof of
Theorem 5 in \cite{OFN12}. 
\end{remark}

We now restate and prove Theorem
\ref{theorem:small-witness-inside-V_1-for-arbitrary-T}.
\thmSmallWitness*
\begin{proof}
  For the sake of contradiction, suppose that the theorem is false for
  some subset $T\subseteq \complement{V_1}$ such that
  $T\cap V_j\neq \emptyset$ for all $j\in \set{2,\ldots, k}$. Our
  proof strategy is to obtain a cheaper $k$-partition than
  $(V_1, \ldots, V_k)$, thereby contradicting the optimality of
  $(V_1, \ldots, V_k)$.  For a subset $X\subseteq V_1$, let
  $(V_{X}, \complement{V_X})$ be the source maximal minimum
  $(X,T)$-terminal cut.

  Among all possible subsets of $V_1$ of size $2k-2$, pick a subset
  $S$ such that $\deltacard(V_{S})$ is maximum. By Lemma \ref{lemma:V_S-inside-V_1} and assumption, we have that $V_S\subsetneq V_1$. By source maximality of the minimum $(S,T)$-terminal cut $(V_S, \complement{V_S})$, we have that $\deltacard(V_S)<\deltacard(V_1)$. 
  Let $u_1,\ldots, u_{2k-2}$ be the vertices in $S$. Since
  $V_{S}\subsetneq V_1$, there exists a vertex
  $u_{2k-1}\in V_1\setminus V_{S}$. Let
  $\core:=\set{u_1,\ldots, u_{2k-1}}=S\cup \set{u_{2k-1}}$. For
  $i\in [2k-1]$, let
  $(B_i, \complement{B_i})$ be the source maximal minimum
  $(\core-\set{u_i},T)$-terminal cut. We note
  that $(B_{2k-1}, \complement{B_{2k-1}})=(V_{S}, \complement{V_{S}})$
  and the size of $\core-\set{u_i}$ is $2k-2$ for every $i\in
  [2k-1]$. By Lemma \ref{lemma:V_S-inside-V_1} and assumption, we have that $B_i\subsetneq V_1$ for every $i\in [2k-1]$. Hence, we have
  \begin{align}
    \deltacard(B_i)&\le \deltacard(V_{S}) <\deltacard(V_1)  \text{ and } B_i\subsetneq V_1 \text{ for every } i\in [2k-1]. \label{statement}
  \end{align}
  The next claim will set us up to apply Theorem \ref{theorem:hypergraph-uncrossing}.
  \begin{claim}\label{claim:u_i-in-B_i}
    For every $i\in [2k-1]$, we have that $u_i\in \complement{B_i}$.
  \end{claim}
  \begin{proof}
    The claim holds for $i=2k-1$ by choice of $u_{2k-1}$. For the sake of contradiction, suppose $u_i\in B_i$ for some $i\in [2k-2]$. Then, the $2$-partition $(V_{S}\cap B_i,\complement{V_{S}\cap B_i})$ is a $(S,T)$-terminal cut and hence
    \[
      \deltacard(V_{S}\cap B_i)\ge \deltacard(V_{S}).
    \]
    We also have that
    \[
      \deltacard(V_{S}\cup B_i)\ge \deltacard(V_{S})
    \]
    since $(V_{S}\cup B_i,\complement{V_S\cup B_i})$ is a $(S,T)$-terminal cut.
    Thus,
    \begin{align*}
      2\deltacard(V_S)
      &\ge \deltacard(V_S) + \deltacard(B_i) \quad \quad \quad \quad \quad \quad \text{ (By choice of $S$)}\\
      &\ge \deltacard(V_S \cup B_i) + \deltacard(V_S \cap B_i) \quad \quad \text{(By submodularity)}\\
    &\ge 2\deltacard(V_S).
    \end{align*}
    Therefore, $\deltacard(V_S)=\deltacard(V_S\cup B_i)$. Moreover,
    $B_i\setminus V_S$ is non-empty since the vertex
    $u_{2k-1}\in B_i\setminus V_S$.  Hence, the $2$-partition
    $(V_S\cup B_i, \complement{V_S\cup B_i})$ is a minimum
    $(S,T)$-terminal cut. However, this contradicts source maximality of the
    minimum $(S,T)$-terminal cut $(V_S, \complement{V_S})$ since
    $u_{2k-1} \in B_i$ and $u_{2k-1} \not \in V_S$. 
  \end{proof}

We note that for every $i\in [2k-1]$, the $2$-partition $(B_i, \complement{B_i})$ is a minimum $(\core-\set{u_i}, \complement{V_1})$-terminal cut since $\complement{V_1}\subseteq \complement{B_i}$. 

  We will now apply Theorem \ref{theorem:hypergraph-uncrossing}.
We consider $U:=V_1$, $R:=\set{u_{2k-1}}\subseteq U$, $S=\set{u_1,\ldots, u_{2k-2}}\subseteq U\setminus R$. Let $p:=2k-2$ and let $(\complement{A_i}, A_i):=(B_i, \complement{B_i})$ for every $i\in [p]$.
The $2$-partition $(\complement{A_i}, A_i)$ is a minimum $((S\cup R)\setminus \set{u_i}, \complement{U})$-terminal cut for every $i\in [p]$.
By Claim \ref{claim:u_i-in-B_i}, we have that $u_i\in A_i$ for every $i\in [p]$. Since $(B_j,\complement{B_j})$ is a $(\core-\set{u_j},T)$-terminal cut, we have that $u_i\not\in \complement{B_j}$ for every distinct $i, j\in [p]$. Thus, $u_i\in A_i\setminus (\cup_{j\in [p]\setminus \set{i}}A_j)$ for every $i\in [p]$.
Therefore, the sets $U$, $R$, $S$ and the $2$-partitions $(\complement{A_i}, A_i)$ for $i\in [p]$ satisfy the conditions of Theorem \ref{theorem:hypergraph-uncrossing}. By Theorem \ref{theorem:hypergraph-uncrossing}, symmetry of the cut function, and statement (\ref{statement}), we obtain a $k$-partition $(P_1, \ldots, P_k)$ of $V$ such that
  \begin{align*}
    \cost(P_1, \ldots, P_k)
    &\le \frac{1}{2}\min\set{\deltacard(A_i) + \deltacard(A_j):i, j\in [p], i\neq j}\\
    &= \frac{1}{2}\min\set{\deltacard(B_i) + \deltacard(B_j): i, j \in [p], i\neq j}\\
    &<\deltacard(V_1)\le \opt.
  \end{align*}
  Thus, we have obtained a $k$-partition whose cost is smaller than $\opt$, a contradiction.

\end{proof}

\begin{remark}
The proof techniques in this section relied only on the submodularity of the hypergraph cut function and the use of Theorem \ref{theorem:hypergraph-uncrossing}. The proof of Theorem \ref{theorem:hypergraph-uncrossing} heavily relied on the function of interest being the hypergraph cut function. As we remarked in Section \ref{section:uncrossing-for-hypergraph-cut-function}, there does not seem to be a counterpart of Theorem \ref{theorem:hypergraph-uncrossing} for submodular functions.
\end{remark}

\section{Structural Theorem for Divide and Conquer}\label{section:structural-theorem-for-DC}
We need a slightly stronger structural theorem to design a faster
algorithm that is based on divide and conquer. We remark again that
the proof techniques in this section will rely only on the
submodularity of the hypergraph cut function and the use of Theorem
\ref{theorem:hypergraph-uncrossing}.

We note that the source maximal minimum $(S,T)$-terminal cut is
identical to the sink minimal minimum $(S,T)$-terminal cut. We define
a $2$-partition $(U,\complement{U})$ to be a \emph{balanced minimum
  $k$-partition split} if there exists a minimum $k$-partition
$(V_1, \ldots, V_k)$ such that $U=\cup_{i=1}^{\lfloor
  k/2\rfloor}V_i$. Since there could be multiple balanced minimum
$k$-partition splits, we will be interested in a minimal balanced
minimum $k$-partition split: a balanced minimum $k$-partition split
$(U, \complement{U})$ is \emph{minimal} if there does not exist
another balanced minimum $k$-partition split $(U', \complement{U'})$
such that $U'$ is strictly contained in $U$.


We need the following two theorems. We defer their proofs to Sections
\ref{subsec:proof-of-theorem-complement-U-side} and
\ref{subsec:proof-of-theorem-U-side} respectively.

\begin{restatable}{theorem}{thmComplementUSide}
  \label{theorem:complement-U-side}
  Let $G=(V,E)$ be a hypergraph and let $\opt$ be the value of a
  minimum $k$-cut in $G$ for some integer $k\ge 2$. Suppose
  $(U,\complement{U})$ is a $2$-partition of $V$ with
  $\deltacard(U)\le \opt$. Then, there exists a subset $S\subseteq U$
  with $|S|\le 2k-2$ such that $(U, \complement{U})$ is the source
  maximal minimum $(S, \complement{U})$-terminal cut in $G$.
\end{restatable}


\begin{restatable}{theorem}{thmUSide}
  \label{theorem:U-side}
  Let $G=(V,E)$ be a hypergraph and let $(U, \complement{U})$ be a
  minimal balanced minimum $k$-partition split in $G$ for some integer
  $k\ge 2$.  Then, for every vertex $u_0\in U$, there exists a subset
  $S\subseteq U\setminus \set{u_0}$ with $|S|\le 2k-3$ such that
  $(U, \complement{U})$ is the \emph{unique} minimum
  $(S\cup \set{u_0}, \complement{U})$-terminal cut in $G$.
\end{restatable}

We now state and prove the structural theorem that facilitates the
faster divide and conquer algorithm.

\begin{theorem}\label{theorem:minimal-balanced-split-recovery}
  Let $G=(V,E)$ be a hypergraph and let $(U, \complement{U})$ be a
  minimal balanced minimum $k$-partition split in $G$ for some integer
  $k\ge 2$.
  Then, for every vertex $u_0\in U$, there exist subsets
  $S\subseteq U\setminus \set{u_0}$ and $T\subseteq \complement{U}$
  with $|S|\le 2k-3$ and $|T|\le 2k-2$ such that $(U, \complement{U})$
  is the source minimal minimum $(S\cup \set{u_0}, T)$-terminal cut in
  $G$.
\end{theorem}
\begin{proof}
  Let $u_0\in U$. Applying Theorem \ref{theorem:U-side} to
  $(U, \complement{U})$ with respect to vertex $u_0\in U$, we obtain a
  set $S\subseteq U$ with $|S|\le 2k-3$ such that
  $(U, \complement{U})$ is the unique minimum
  $(S\cup \set{u_0}, \complement{U})$-terminal cut in $G$.

  Applying Theorem \ref{theorem:complement-U-side} to
  $(\complement{U}, U)$, we obtain a set $T\subseteq \complement{U}$
  with $|T|\le 2k-2$ 
  such that
    $(\complement{U},U)$ is source-maximal minimum $(T, U)$ cut in
    $G$. Hence, by interchanging source and sink, $(U,\complement{U})$
    is the source-minimal minimum $(U,T)$ cut in $G$.
  
  We will show that $(U, \complement{U})$ is the source minimal
  minimum $(S\cup \set{u_0}, T)$-terminal cut in $G$. We first show
  that $(U, \complement{U})$ is a minimum
  $(S\cup \set{u_0}, T)$-terminal cut. Let $(X, \complement{X})$ be a
  minimum $(S\cup \set{u_0}, T)$-terminal cut.  Then,
  \[
    \deltacard(U)\ge \deltacard(X)
  \]
  since $(U, \complement{U})$ is a $(S\cup \set{u_0}, T)$-terminal cut.
  Since $(X\cap U, \complement{X\cap U})$ is a $(S\cup \set{u_0}, \complement{U})$-terminal cut, we have
  \[
    \deltacard(X\cap U)\ge \deltacard(U).
  \]
  Since $(X\cup U, \complement{X\cup U})$ is a $(U, T)$-terminal cut, we have
  \[
    \deltacard(X\cup U)\ge \deltacard(U).
  \]
  The above three inequalities in conjunction with the submodularity of the cut function imply that
  \[
    2\deltacard(U) \ge \deltacard(X) + \deltacard(U) \ge \deltacard(X\cap U) + \deltacard(X\cup U) \ge 2\deltacard(U).
  \]
  Hence, all the above inequalities should be equations and therefore, $\deltacard(U)=\deltacard(X)$.

  Next, we show that $(U, \complement{U})$ is the source minimal
  minimum $(S\cup \set{u_0}, T)$-terminal cut. For the sake of
  contradiction, suppose $(X, \complement{X})$ is the source minimal
  minimum $(S\cup \set{u_0}, T)$-terminal cut with $X\neq U$. We have
  the following cases.

\noindent \textbf{Case 1.} Suppose $X\supsetneq U$. Then, $(U, \complement{U})$ contradicts source minimality of the  minimum $(S\cup \set{u_0}, T)$-terminal cut $(X, \complement{X})$.

\noindent \textbf{Case 2.} Suppose $X\subsetneq U$. Then,
$(X, \complement{X})$ is also a minimum $(S\cup \set{u_0},
\complement{U})$-terminal cut, 
a contradiction since the choice of $S$
  implies that $(U,\complement{U})$ is unique minimum $(S\cup \set{u_0},
\complement{U})$-terminal cut.

\noindent \textbf{Case 3.} Suppose $X\setminus U\neq \emptyset$ and $X\setminus \complement{U}\neq \emptyset$. Then, we have
\[
\deltacard(X\cap U)\ge \deltacard(X)
\]
since $(X\cap U, \complement{X\cap U})$ is a $(S\cup \set{u_0}, T)$-cut. We also have
\[
\deltacard(X\cup U)\ge \deltacard(X)
\]
since $(X\cup U, \complement{X\cup U})$ is a $(S\cup \set{u_0}, T)$-cut. The above two  inequalities in conjunction with the submodularity of the cut function imply that
\[
2\deltacard(X) = \deltacard(X) + \deltacard(U) \ge \deltacard(X\cap U) + \deltacard(X\cup U) \ge 2\deltacard(X).
\]
Therefore, $\deltacard(X\cap U)=\deltacard(X)$. Thus, the $2$-partition $(X\cap U, \complement{X\cap U})$ contradicts source minimality of the minimum $(S\cup \set{u_0}, T)$-terminal cut $(X, \complement{X})$.
\end{proof}

\subsection{Proof of Theorem \ref{theorem:complement-U-side}}\label{subsec:proof-of-theorem-complement-U-side}
We restate and prove Theorem \ref{theorem:complement-U-side} in this section.
\thmComplementUSide*
\begin{proof}
For the sake of contradiction, suppose that the theorem is false. Our proof strategy is to obtain a cheaper $k$-partition with cost strictly less than $\opt$, thereby contradicting optimality.
For a subset $X\subseteq U$, let $(V_{X}, \complement{V_X})$ be the source maximal minimum $(X,\complement{U})$-terminal cut.

Let $X$ be an arbitrary subset of $U$ with $|X|=2k-2$. Since we are assuming that the theorem is false, it follows that $V_X\neq U$. By definition, we have that $V_X\subsetneq U$. By source maximality of the minimum $(X,\complement{U})$-terminal cut $(V_X,\complement{V_X})$, we have that $\deltacard(V_X)<\deltacard(U)$.


Among all possible subsets of $U$ of size $2k-2$, pick a subset $S$ such that $\deltacard(V_{S})$ is maximum. Then, $V_S\subsetneq U$ and
\[\deltacard(V_X)\le \deltacard(V_{S}) <\deltacard(U) \text{ for every $X\subseteq U$ with $|X|=2k-2$}. \]

The rest of the proof is identical to the proof of Theorem \ref{theorem:small-witness-inside-V_1-for-arbitrary-T}.
Let $u_1,\ldots, u_{2k-2}$ be the vertices in $S$. Since $V_{S}\subsetneq U$, there exists a vertex $u_{2k-1}\in U\setminus V_{S}$. Let $\core:=\set{u_1,\ldots, u_{2k-1}}=S\cup \set{u_{2k-1}}$. Also, let $(B_i, \complement{B_i})$ be the source maximal minimum $(\core-\set{u_i},\complement{U})$-terminal cut for every $i\in [2k-1]$. We note that $(B_{2k-1}, \complement{B_{2k-1}})=(V_{S}, \complement{V_{S}})$ and the size of $\core-\set{u_i}$ is $2k-2$ for every $i\in [2k-1]$. Hence, we have
\begin{align}
    \deltacard(B_i)&\le \deltacard(V_{S}) <\deltacard(U)  \text{ and } B_i\subsetneq U \text{ for every } i\in [2k-1]. \label{statement:complement-U-side}
\end{align}
The next claim will set us up to apply Theorem \ref{theorem:hypergraph-uncrossing}.
\begin{claim}\label{claim:u_i-in-B_i-for-D-C}
For every $i\in [2k-1]$, we have that $u_i\in \complement{B_i}$.
\end{claim}
\begin{proof}
The claim holds for $i=2k-1$ by choice of $u_{2k-1}$. For the sake of contradiction, suppose $u_i\in B_i$ for some $i\in [2k-2]$. Then, the $2$-partition $(V_{S}\cap B_i,\complement{V_{S}\cap B_i})$ is a $(S,\complement{U})$-terminal cut and hence
\[
\deltacard(V_{S}\cap B_i)\ge \deltacard(V_{S}).
\]
We also have
\[
\deltacard(V_{S}\cup B_i)\ge \deltacard(V_{S})
\]
since $(V_{S}\cup B_i,\complement{V_S\cup B_i})$ is a $(S,\complement{U})$-terminal cut.
Thus,
\begin{align*}
    2\deltacard(V_S)
    &\ge \deltacard(V_S) + \deltacard(B_i) \quad \quad \quad \quad \quad \quad \text{  (By choice of $S$)}\\
    &\ge \deltacard(V_S \cup B_i) + \deltacard(V_S \cap B_i) \quad \quad \text{(By submodularity)}\\
    &\ge 2\deltacard(V_S).
\end{align*}
Therefore, $\deltacard(V_S)=\deltacard(V_S\cup B_i)$. Moreover, $B_i\setminus V_S$ is non-empty since the vertex $u_{2k-1}\in B_i\setminus V_S$.
Hence, the $2$-partition $(V_S\cup B_i, \complement{V_S\cup B_i})$ is a minimum $(S,\complement{U})$-terminal cut and it contradicts source maximality of the minimum $(S,\complement{U})$-terminal cut $(V_S, \complement{V_S})$.
\end{proof}

Let $p:=2k-2$. Using Claim \ref{claim:u_i-in-B_i-for-D-C}, we observe that the sets $U$, $R:=\set{u_{2k-1}}$, $S=\set{u_1, \ldots, u_{2k-2}}$, and the partitions $(\complement{A_i}, A_i):=(B_i, \complement{B_i})$ for $i\in [p]$ satisfy the conditions of Theorem \ref{theorem:hypergraph-uncrossing}. By Theorem \ref{theorem:hypergraph-uncrossing}, symmetry of the cut function, and statement (\ref{statement:complement-U-side}), we obtain a $k$-partition $(P_1, \ldots, P_k)$ of $V$ such that
\begin{align*}
\cost(P_1, \ldots, P_k)
&\le \frac{1}{2}\min\set{\deltacard(A_i) + \deltacard(A_j):i, j\in [p], i\neq j}\\
&= \frac{1}{2}\min\set{\deltacard(B_i) + \deltacard(B_j): i, j \in [p], i\neq j}\\
&<\deltacard(U)\le \opt.
\end{align*}
Thus, we have obtained a $k$-partition whose cost is smaller than $\opt$, a contradiction.

\end{proof}

\subsection{Proof of Theorem \ref{theorem:U-side}}\label{subsec:proof-of-theorem-U-side}
We restate and prove Theorem \ref{theorem:U-side} in this section.
\thmUSide*
\begin{proof}
Let $u_0\in U$ and let $\opt$ be the value of a minimum $k$-cut in $G$.
Consider the collection
\[
\collection :=\set{Q\subseteq V\setminus \set{u_0}:\ \complement{U}\subsetneq Q,\ \deltacard(Q)\le \deltacard(U)}.
\]
Let $S$ be an inclusion-wise minimal subset of $U\setminus \set{u_0}$ such that $S\cap Q\neq \emptyset$ for all $S\in \collection$ i.e., the set $S$ is completely contained in $U\setminus \set{u_0}$ and is a minimal transversal of the collection $\collection$. Proposition \ref{prop:unique-min} and Lemma \ref{lemma:transversal-size} complete the proof of the theorem for this choice of $S$.
\end{proof}

\begin{proposition}\label{prop:unique-min}
The $2$-partition $(U, \complement{U})$ is the unique minimum $(S\cup \set{u_0}, \complement{U})$-terminal cut in $G$.
\end{proposition}
\begin{proof}
For the sake of contradiction, suppose $(X, \complement{X})$ is a minimum $(S\cup \set{u_0}, \complement{U})$-terminal cut in $G$ such that $X\neq U$. Then, $\deltacard(\complement{X})\le \deltacard(U)$ since $(U,\complement{U})$ is a feasible $(S\cup \set{u_0}, \complement{U})$-terminal cut. By definition, $\complement{U}\subsetneq \complement{X}\subseteq V\setminus \set{u_0}$. Hence, the set $\complement{X}$ is in the collection $\collection$. Since $S$ is a transversal of the collection $\collection$, we have that $S\cap \complement{X}\neq \emptyset$. This contradicts the fact that $S$ is contained in $X$.
\end{proof}

\begin{lemma}\label{lemma:transversal-size}
The size of the transversal $S$ is at most $2k-3$.
\end{lemma}
\begin{proof}
For the sake of contradiction, suppose $|S|\ge 2k-2$. We will construct a balanced minimum $k$-partition split in $G$ that contradicts the minimality of the balanced minimum $k$-partition split $(U, \complement{U})$.
Let $S=\set{u_1, \ldots, u_p}$ for $p\ge 2k-2$. For each $i\in [p]$, let $(\complement{A_i}, A_i)$ be the source minimal minimum $((S\cup \set{u_0})\setminus \set{u_i}, \complement{U})$-terminal cut.

\begin{claim}\label{claim:min-separating-cuts-miss-u_i}
For every $i\in [p]$, we have that $\deltacard(A_i)\le \deltacard(U)$ and $u_i\in A_i$.
\end{claim}
\begin{proof}
Let $i\in [p]$. Since $S$ is a minimal transversal for the collection $\collection$, there exists a set $B_i\in \collection$ such that $B_i \cap S = \set{u_i}$. Hence, $(\complement{B_i}, B_i)$ is a feasible $(S\cup \set{u_0}\setminus \set{u_i}, \complement{U})$-terminal cut. Therefore,
\[
\deltacard(A_i) \le \deltacard(B_i)\le \deltacard(U).
\]

We will show that $A_i$ is in the collection $\collection$.
By definition, $A_i\subseteq V\setminus \set{u_0}$ and $\complement{U}\subseteq A_i$. If $\complement{U}=A_i$, then the above inequalities are equations implying that $(B_i, \complement{B_i})$ is a minimum $((S\cup \set{u_0})\setminus \set{u_i}, \complement{U})$-terminal cut, and consequently, $(B_i, \complement{B_i})$ contradicts source minimality of the minimum $((S\cup \set{u_0})\setminus \set{u_i}, \complement{U})$-terminal cut $(\complement{A_i}, A_i)$. Therefore, $\complement{U}\subsetneq A_i$. Hence, $A_i$ is in the collection $\collection$.

We recall that the set $S$ is a transversal for the collection $\collection$ and none of the elements of $S\setminus \set{u_i}$ are in $A_i$. Hence, the element $u_i$ must be in $A_i$.
\end{proof}

Using Claim \ref{claim:min-separating-cuts-miss-u_i}, we observe that the sets $U$, $R:=\set{u_0}$, $S$, and the partitions $(\complement{A_i}, A_i)$ for $i\in [p]$ satisfy the conditions of Theorem \ref{theorem:hypergraph-uncrossing}. By Theorem \ref{theorem:hypergraph-uncrossing} and Claim \ref{claim:min-separating-cuts-miss-u_i}, we obtain $k$-partition $(P_1, \ldots, P_k)$ of $V$ such that $\complement{U}\subsetneq P_k$ and
\begin{align*}
\cost(P_1, \ldots, P_k)
&\le \frac{1}{2}\min\set{\deltacard(A_i) + \deltacard(A_j):i, j\in [p], i\neq j}
\le \deltacard(U)
\le \opt.
\end{align*}
Thus, we have obtained a minimum $k$-partition $(P_1, \ldots, P_k)$ such that $\complement{U}\subsetneq P_k$. Now, consider $U':=\cup_{i=1}^{\lfloor k/2 \rfloor} P_i$. We observe that $(U', \complement{U'})$ is a balanced minimum $k$-partition split such that $U'$ is strictly contained in $U$, a contradiction to minimality of the balanced minimum $k$-partition split $(U, \complement{U})$.

\end{proof}

\section{Divide and Conquer Algorithm}
\label{section:DC-algo}
In this section, we design an $n^{O(k)}$-time algorithm based on
divide and conquer. We describe the algorithm in Figure
\ref{fig:k-cut-divide-and-conquer-algorithm} and its run-time
guarantee in Theorem \ref{theorem:min-k-cut-divide-and-conquer-algo}.
To recap from the introduction, the high-level idea is to use minimum $(S,T)$-terminal cuts to find a balanced minimum $k$-partition split
$(U,\complement{U})$; the balance helps in cutting the recursion depth
which results in savings in the overall run-time. 


\begin{figure*}[ht]
\centering\small
\begin{algorithm}
\textul{Algorithm DIVIDE-AND-CONQUER-CUT$(G,k)$}\+
\\  {\bf Input: } Hypergraph $G=(V,E)$ and an integer $k\ge 1$
\\  {\bf Output: } A $k$-partition corresponding to a minimum $k$-cut in $G$
\\  If $k=1$\+
\\      Return $V$\-
\\  Initialize $\mathcal{R}\leftarrow \emptyset$ and $p\leftarrow \lfloor k/2\rfloor $
\\  For every disjoint $S, T\subset V$ with $|S|, |T|\le 2k-2$\+
\\      Compute the source minimal minimum $(S, T)$-terminal cut $(U,\complement{U})$
\\      If $|U|\ge p$ and $\complement{U}\ge k-p$\+
\\          $\mathcal{R}\leftarrow \mathcal{R}\cup \set{(U, \complement{U})}$
\\          $\mathcal{P}_U:= \text{DIVIDE-AND-CONQUER-CUT}(G[U],p)$
\\          $\mathcal{P}_{\complement{U}}:= \text{DIVIDE-AND-CONQUER-CUT}(G[\complement{U}],k-p)$
\\          $C_U:=$ Partition of $V$ obtained by concatenating the parts in $\mathcal{P}_U$ and $\mathcal{P}_{\complement{U}}$\-\-
\\  Among all $k$-partitions $C_U$ with $(U,\complement{U})\in \mathcal{R}$, pick the one with minimum cost and return it \-
\end{algorithm}
\caption{Divide and conquer algorithm to compute minimum $k$-cut in hypergraphs.}
\label{fig:k-cut-divide-and-conquer-algorithm}
\end{figure*}

\begin{theorem}\label{theorem:min-k-cut-divide-and-conquer-algo}
  Let $G=(V,E)$ be a $n$-vertex hypergraph of size $p$ and let $k$ be
  an integer. Then, algorithm DIVIDE-AND-CONQUER-CUT$(G,k)$ returns a
  partition corresponding to a minimum $k$-cut in $G$ and it can be
  implemented to run in $O(n^{8k}T(n, p))$ time, where $T(n,p)$
  denotes the time complexity for computing the source minimal minimum $(S,T)$-terminal
  cut in a $n$-vertex hypergraph of size $p$.
\end{theorem}
\begin{proof}
  We first show the correctness of the algorithm. All candidates
  considered by the algorithm correspond to a $k$-partition, so we
  only have to show that the algorithm returns a $k$-partition
  corresponding to a minimum $k$-cut.  We show this by induction on
  $k$. The base case of $k=1$ is trivial. We show the induction step.
  Let $(P_1, \ldots, P_k)$ be a minimum $k$-partition in $G$ such that
  for $p=\lfloor k/2 \rfloor$, the $2$-partition
  $(U_0:=\cup_{i=1}^p P_i, \complement{U_0}=\cup_{i=p+1}^k P_i)$ is a
  minimal balanced minimum $k$-partition split.  Let $\opt$ denote the
  value of a minimum $k$-partition in $G$.

  We observe that $|U_0|\ge p$ and $|\complement{U_0}|\ge k-p$.  By
  Theorem \ref{theorem:minimal-balanced-split-recovery}, the
  $2$-partition $(U_0, \complement{U_0})$ is in $\mathcal{R}$.  By
  induction hypothesis, the algorithm will return a $p$-partition
  $\mathcal{P}_{U_0}=(Q_1, \ldots, Q_p)$ of $U_0$ and a $(k-p)$-partition
  $\mathcal{P}_{\complement{U_0}}=(Q_{p+1}, \ldots, Q_k)$ of
  $\complement{U_0}$ such that
  \begin{align*}
    \cost_{G[U_0]}(Q_1, \ldots, Q_p)&\le \cost_{G[U_0]}(P_1, \ldots, P_p) \text{ and }\\
    \cost_{G[\complement{U_0}]}(Q_{p+1}, \ldots, Q_k)&\le \cost_{G[U_0]}(P_{p+1}, \ldots, P_k).
  \end{align*}
  Hence, the cost of the partition $(Q_1, \ldots, Q_k)$ returned by
  the algorithm is
  \begin{align*}
    \deltacard(U_0)+\cost_{G[U_0]}(Q_1, \ldots, Q_p)&+\cost_{G[\complement{U_0}]}(Q_{p+1}, \ldots, Q_k) \\
                                                &\le\deltacard(U_0)+\cost_{G[U_0]}(P_1, \ldots, P_p)+\cost_{G[\complement{U_0}]}(P_{p+1}, \ldots, P_k) \\
                                                &= \cost_G(P_1, \ldots, P_k) \\
                                                &= \opt.
  \end{align*}

  Next, we prove the run-time bound. We will derive an upper bound
  $N(k, n)$ on the number of source minimal minimum $(S,T)$-terminal cut computations
  executed by the algorithm, where we assume that $N(k, n)$ is an
  increasing function of $k$ and $n$. We know that $N(1, n)=O(1)$. We
  have
  \begin{align*}
    N(k, n) = O\left( n^{4k-4}\right) \left(1+N\left(\left\lceil \frac{k}{2} \right\rceil, n\right) + N\left(\left\lfloor \frac{k}{2} \right\rfloor, n\right)\right).
  \end{align*}
  By substitution, it can be verified that $N(k, n)=O(n^{8k})$.  The
  running time is dominated by the number of terminal cut
  computations and this yields the desired time bound on the algorithm.

\end{proof}

\paragraph{Acknowledgements.} We thank Krist\'{o}f B\'{e}rczi, Tam\'{a}s Kir\'{a}ly, and Chao Xu 
for helpful discussions on alternative approaches for deterministic hypergraph $k$-cut during preliminary stages of this work.
We also thank Sagemath (\url{www.sagemath.org}) and CoCalc (\url{www.cocalc.com}) 
for providing software platforms to conduct experiments that led us towards our main structural theorem. 

\bibliographystyle{amsplain}
\bibliography{references}

\providecommand{\bysame}{\leavevmode\hbox to3em{\hrulefill}\thinspace}
\providecommand{\MR}{\relax\ifhmode\unskip\space\fi MR }
\providecommand{\MRhref}[2]{%
  \href{http://www.ams.org/mathscinet-getitem?mr=#1}{#2}
}
\providecommand{\href}[2]{#2}
\begin{thebibliography}{10}

\bibitem{BSW19}
N.~Buchbinder, R.~Schwartz, and B.~Weizman, \emph{A simple algorithm for the
  multiway cut problem}, Operations Research Letters \textbf{47} (2019), no.~6,
  587--593.

\bibitem{CXY19}
K.~Chandrasekaran, C.~Xu, and X.~Yu, \emph{Hypergraph $k$-cut in randomized
  polynomial time}, Mathematical Programming (Preliminary version in SODA 2018)
  (2019).

\bibitem{CE11}
C.~Chekuri and A.~Ene, \emph{{Approximation Algorithms for Submodular Multiway
  Partition}}, Proceedings of the 52nd IEEE Annual Symposium on Foundations of
  Computer Science, FOCS, 2011, pp.~807--816.

\bibitem{CL15}
C.~Chekuri and S.~Li, \emph{{A note on the hardness of approximating the
  $k$-way Hypergraph Cut problem}}, Manuscript,
  \url{http://chekuri.cs.illinois.edu/papers/hypergraph-kcut.pdf}, 2015.

\bibitem{CQX19}
C.~Chekuri, K.~Quanrud, and C.~Xu, \emph{{LP Relaxation and Tree Packing for
  Minimum $k$-cuts}}, 2nd Symposium on Simplicity in Algorithms, SOSA, 2019,
  pp.~7:1--7:18.

\bibitem{ChekuriX18}
C.~Chekuri and C.~Xu, \emph{Minimum cuts and sparsification in hypergraphs},
  SIAM Journal on Computing \textbf{47} (2018), no.~6, 2118--2156.

\bibitem{DEFPR03}
R.~Downey, V.~Estivill-Castro, M.~Fellows, E.~Prieto, and F.~Rosamund,
  \emph{Cutting up is hard to do: The parameterised complexity of k-cut and
  related problems}, Electronic Notes in Theoretical Computer Science
  \textbf{78} (2003), 209--222.

\bibitem{EneVW13}
A.~Ene, J.~Vondr{\'a}k, and Y.~Wu, \emph{Local distribution and the symmetry
  gap: Approximability of multiway partitioning problems}, Proceedings of the
  24th annual ACM-SIAM symposium on Discrete algorithms, SODA, 2013,
  pp.~306--325.

\bibitem{FPZ19}
K.~Fox, D.~Panigrahi, and F.~Zhang, \emph{Minimum cut and minimum $k$-cut in
  hypergraphs via branching contractions}, Proceedings of the 30th Annual
  ACM-SIAM Symposium on Discrete Algorithms, SODA, 2019, pp.~881--896.

\bibitem{F10}
T.~Fukunaga, \emph{Computing minimum multiway cuts in hypergraphs}, Discrete
  Optimization \textbf{10} (2013), no.~4, 371--382.

\bibitem{GKP17}
M.~Ghaffari, D.~Karger, and D.~Panigrahi, \emph{Random contractions and
  sampling for hypergraph and hedge connectivity}, Proceedings of the 28th
  Annual ACM-SIAM Symposium on Discrete Algorithms, SODA, 2017, p.~1101–1114.

\bibitem{GH94}
O.~Goldschmidt and D.~Hochbaum, \emph{{A Polynomial Algorithm for the $k$-cut
  Problem for Fixed $k$}}, Mathematics of Operations Research \textbf{19}
  (1994), no.~1, 24--37.

\bibitem{GH88}
O.~{Goldschmidt} and D.~S. {Hochbaum}, \emph{Polynomial algorithm for the k-cut
  problem}, Proceedings of the 29th Annual Symposium on Foundations of Computer
  Science, 1988, pp.~444--451.

\bibitem{GQ}
F.~Guinez and M.~Queyranne, \emph{The size-constrained submodular k-partition
  problem}, Unpublished manuscript. Available at
  \url{https://docs.google.com/viewer?a=v&pid=sites&srcid=ZGVmYXVsdGRvbWFpbnxmbGF2aW9ndWluZXpob21lcGFnZXxneDo0NDVlMThkMDg4ZWRlOGI1}.
  See also
  \url{https://smartech.gatech.edu/bitstream/handle/1853/43309/Queyranne.pdf},
  2012.

\bibitem{GLL20-STOC}
A.~Gupta, E.~Lee, and J.~Li, \emph{The {Karger-Stein Algorithm is Optimal for}
  $k$-cut}, Preprint in arXiv: 1911.09165, 2019.

\bibitem{KYN07}
Y.~Kamidoi, N.~Yoshida, and H.~Nagamochi, \emph{{A Deterministic Algorithm for
  Finding All Minimum $k$-Way Cuts}}, SIAM Journal on Computing \textbf{36}
  (2007), no.~5, 1329--1341.

\bibitem{Karger00}
D.~Karger, \emph{Minimum cuts in near-linear time}, Journal of the ACM
  \textbf{47} (2000), no.~1, 46--76.

\bibitem{KS96}
D.~Karger and C.~Stein, \emph{A new approach to the minimum cut problem},
  Journal of the ACM \textbf{43} (1996), no.~4, 601--640.

\bibitem{KT11}
K.~Kawarabayashi and M.~Thorup, \emph{The minimum k-way cut of bounded size is
  fixed-parameter tractable}, Proceedings of the 52nd Annual Symposium on
  Foundations of Computer Science, FOCS, 2011, pp.~160--169.

\bibitem{La73}
E.~Lawler, \emph{{Cutsets and Partitions of Hypergraphs}}, {Networks}
  \textbf{3} (1973), 275--285.

\bibitem{Li19}
J.~Li, \emph{Faster minimum k-cut of a simple graph}, Proceedings of the 60th
  Annual Symposium on Foundations of Computer Science, FOCS, 2019,
  pp.~1056--1077.

\bibitem{Ma17dks}
P.~Manurangsi, \emph{{Almost-polynomial Ratio ETH-hardness of Approximating
  Densest $k$-subgraph}}, Proceedings of the 49th Annual ACM Symposium on
  Theory of Computing, STOC, 2017, pp.~954--961.

\bibitem{Ma17}
\bysame, \emph{{Inapproximability of Maximum Biclique Problems, Minimum $k$-Cut
  and Densest At-Least-$k$-Subgraph from the Small Set Expansion Hypothesis}},
  Proceedings of the 44th International Colloquium on Automata, Languages, and
  Programming, ICALP, 2017, pp.~79:1--79:14.

\bibitem{M06}
D.~Marx, \emph{Parameterized graph separation problems}, Theoretical Computer
  Science \textbf{351} (2006), no.~3, 394--406.

\bibitem{NI92}
H.~Nagamochi and T.~Ibaraki, \emph{A linear-time algorithm for finding a sparse
  $k$-connected spanning subgraph of a $k$-connected graph}, Algorithmica
  \textbf{7} (1992), no.~1-6, 583--596.

\bibitem{OFN12}
K.~Okumoto, T.~Fukunaga, and H.~Nagamochi, \emph{Divide-and-conquer algorithms
  for partitioning hypergraphs and submodular systems}, Algorithmica
  \textbf{62} (2012), no.~3, 787--806.

\bibitem{Q19}
K.~Quanrud, \emph{{Fast and Deterministic Approximations for k-Cut}},
  Proceedings of Approximation, Randomization, and Combinatorial Optimization.
  Algorithms and Techniques, APPROX, 2019, pp.~23:1--23:20.

\bibitem{Q99}
M.~Queyranne, \emph{On optimum size-constrained set partitions}, 1999, Talk at
  Aussiois workshop on Combinatorial Optimization.

\bibitem{RS10}
P.~Raghavendra and D.~Steurer, \emph{{Graph Expansion and the Unique Games
  Conjecture}}, Proceedings of the 42nd Annual ACM Symposium on Theory of
  Computing, STOC, 2010, pp.~755--764.

\bibitem{SV95}
H.~Saran and V.~Vazirani, \emph{{Finding k Cuts within Twice the Optimal}},
  SIAM Journal on Computing \textbf{24} (1995), no.~1, 101--108.

\bibitem{SW97}
M.~Stoer and F.~Wagner, \emph{A simple min-cut algorithm}, Journal of the ACM
  (JACM) \textbf{44} (1997), no.~4, 585--591.

\bibitem{Th08}
M.~Thorup, \emph{{Minimum $k$-way Cuts via Deterministic Greedy Tree Packing}},
  Proceedings of the 40th Annual ACM Symposium on Theory of Computing, STOC,
  2008, pp.~159--166.

\bibitem{Xi08}
M.~Xiao, \emph{{An Improved Divide-and-Conquer Algorithm for Finding All
  Minimum k-Way Cuts}}, Proceedings of 19th International Symposium on
  Algorithms and Computation, ISAAC, 2008, pp.~208--219.

\bibitem{ZNI05}
L.~Zhao, H.~Nagamochi, and T.~Ibaraki, \emph{{Greedy splitting algorithms for
  approximating multiway partition problems}}, Mathematical Programming
  \textbf{102} (2005), no.~1, 167--183.

\end{thebibliography}

\end{document}